%% file: main.tex
\documentclass[11pt]{article}

\input{macros.tex}

\title{Fully-Dynamic Submodular Cover with Bounded Recourse}
\author{Anupam Gupta\thanks{Computer Science Department, Carnegie Mellon
		University, Pittsburgh, PA 15213. This research was done under the
auspices of the Indo-US Virtual Networked Joint Center IUSSTF/JC-017/2017.
Research supported in part by NSF awards
  CCF-1907820, CCF1955785, and CCF-2006953. Emails:
		\texttt{\{anupamg,roiel\}@cs.cmu.edu}. }
	\and
	Roie Levin$^{*}$
}
\date{}

\begin{document}

\maketitle

\begin{abstract}
   In submodular covering problems, we are given a monotone,
   nonnegative submodular function $f: 2^\unvrs \rightarrow \R_+$ and
   wish to find the min-cost set $S \subseteq \unvrs$ such that $f(S)
   = f(\unvrs)$. When $f$ is a coverage function, this captures
   \setcov as a special case. We introduce a general framework for
   solving such problems in a fully-dynamic setting where the function
   $f$ changes over time, and only a bounded number of updates to the
   solution (a.k.a.\ recourse) is allowed. For concreteness, suppose a nonnegative monotone submodular
   integer-valued function $g_t$ is added or removed from an active set $G^{(t)}$ at
   each time $t$. If $f^{(t)} = \sum_{g \in G^{(t)}} g$ is the sum of
   all active functions, we wish to maintain a competitive solution to
   \subcov for $f^{(t)}$ as this active set changes, and with low
   recourse. For example, if each $g_t$ is the (weighted) rank
   function of a matroid, we would be dynamically maintaining a
   low-cost common spanning set for a changing collection of matroids.
	
   We give an algorithm that maintains an
   $O(\log (\fmax / \fmin))$-competitive solution, where
   $\fmax, \fmin$ are the largest/smallest marginals of $f^{(t)}$. The
   algorithm guarantees a total recourse of
   $O(\log (\cmax / \cmin) \cdot \sum_{t\leq T} g_t(\unvrs))$, where
   $\cmax,\cmin$ are the largest/smallest costs of elements in
   $\unvrs$. This competitive ratio is best possible even in the
   offline setting, and the recourse bound is optimal up to the
   logarithmic factor. For monotone submodular functions that also
   have positive mixed third derivatives, we show an optimal recourse
   bound of $O( \sum_{t\leq T} g_t(\unvrs))$. This structured
   class includes set-coverage functions, so our algorithm matches
   the known $O(\log n)$-competitiveness and $O(1)$ recourse
   guarantees for fully-dynamic \setcov. Our work simultaneously simplifies and unifies previous results, as
   well as generalizes to a significantly larger class of covering
   problems. Our key technique is a new potential function inspired by
   Tsallis entropy. We also extensively use the idea of \textit{Mutual
     Coverage}, which generalizes the classic notion of mutual
   information.
\end{abstract}

\vfill

\thispagestyle{empty}

\pagebreak 

\setcounter{page}{1}

\input{intro.tex}

\input{warmupI.tex}
\input{generalcosts.tex}
\input{specialinstances.tex}

\section*{Note on Changes since Publication}
\label{sec:version_notes}
In the original version of this paper that appeared in FOCS 2020, \cref{lem:phihalf_increase} contained an error: the potential $\Phi_{1/2}(\pi)$ from \cref{sec:3incr} could suffer unbounded increase with the addition of new functions. The issue is resolved with Assumption \ref{assum:singleton_set} (which is without loss of generality) and the minor changes to the algorithm at the end of \cref{subsec:3inc_algo}. The proof of \cref{lem:phihalf_increase} has also been updated and simplified.

\appendix

{\footnotesize
	\bibliography{refs}
	\bibliographystyle{alpha}
}

\input{freqalgo.tex}

\input{combiner.tex}
\input{beyondsubmodularity.tex}

\input{appendix.tex}

\end{document}

%% file: macros.tex
\usepackage{fullpage}
\usepackage[T1]{fontenc}
\usepackage{bbm}
\usepackage{framed} 
\usepackage{url}
\usepackage{complexity}
\usepackage{booktabs}
\usepackage{amsmath,amssymb}
\usepackage{amsthm}
\usepackage{nicefrac}
\usepackage{microtype}
\usepackage{calc}
\usepackage{enumerate}
\usepackage{enumitem}
\usepackage[usenames,dvipsnames]{xcolor}
\usepackage[colorlinks,citecolor=blue,linkcolor=BrickRed]{hyperref}
\usepackage{algorithm}
\usepackage{graphicx}
\usepackage[font=footnotesize,labelfont=bf]{subcaption}
\usepackage[font=footnotesize,labelfont=bf]{caption}
\usepackage[nobreak=true]{mdframed}
\usepackage{appendix}
\usepackage[noend]{algpseudocode}
\usepackage[colorinlistoftodos]{todonotes}
\usepackage{xr}
\usepackage{array}
\usepackage{xspace}
\usepackage{hyperref}
\usepackage[capitalise]{cleveref}
\crefrangelabelformat{enumi}{#3#1#4--#5#2#6}
\usepackage{parskip}
\usepackage{chngcntr}
\usepackage{mathtools}

\usepackage{nameref}

\usepackage{tikz}
\usetikzlibrary{patterns}

\makeatletter
 \allowdisplaybreaks
\makeatother

\DeclareMathOperator*{\expectation}{\mathbb{E}}
\let\poly\relax
\DeclareMathOperator*{\poly}{poly}

\DeclareMathOperator*{\probability}{\mathbb{P}}

\renewcommand\R{\mathbb{R}}
\newcommand\Z{\mathbb{Z}}
\newcommand\mff{\mathfrak{F}}
\newcommand\mfI{\mathfrak{I}}

\newcommand\eps{\epsilon}

\newcommand\unvrs{\mathcal{N}}

\newcommand\expect[1]{\expectation\left[ #1 \right]}

\newcommand\prob[1]{\probability\left( #1 \right)}

\newcommand\mutcov{\mathcal{I}}
\newcommand\opt{\textsc{Opt}}
\newcommand\sse{\subseteq}
\newcommand\defeq{ \xspace \stackrel{\tiny \mathclap{\mbox{(def)}}}{=} \xspace}

\counterwithin{equation}{section}

\newcommand\sscText{$\ComplexityFont{SSC}$\xspace}
\newcommand\osscText{$\ComplexityFont{FDSSC}$\xspace}

\usepackage{eqparbox}
\renewcommand{\algorithmiccomment}[1]{\hfill\eqparbox{COMMENT}{// #1}}

\newcommand\fmin{f_{\min}}
\newcommand\fmax{f_{\max}}

\newcommand\cmin{c_{\min}}
\newcommand\cmax{c_{\max}}

\newcommand\Ell{\mathcal{L}}

\newcommand\mst{\textsc{MinimumSpanningTree}\xspace}
\newcommand\mstt{\textsc{MinimumSteinerTree}\xspace}
\newcommand\setcov{\textsc{SetCover}\xspace}
\newcommand\hitset{\textsc{HittingSet}\xspace}
\newcommand\hvc{\textsc{HypergraphVertexCover}\xspace}
\newcommand\subcov{\textsc{SubmodularCover}\xspace}

\theoremstyle{plain}

\newtheorem{theorem}{Theorem}[section]
\newtheorem{definition}{Definition}[section]

\newtheorem{lemma}{Lemma}[section]
\newtheorem{fact}{Fact}[section]
\newtheorem{observation}{Observation}
\newtheorem{claim}{Claim}[section]
\newtheorem{corollary}{Corollary}[section]

\newtheorem{assumption}{Assumption}[section]

\newlength{\continueindent}
\usepackage{etoolbox}
\makeatletter
\newcommand*{\ALG@customparshape}{\parshape 2 \leftmargin \linewidth \dimexpr\ALG@tlm+\continueindent\relax \dimexpr\linewidth+\leftmargin-\ALG@tlm-\continueindent\relax}
\apptocmd{\ALG@beginblock}{\ALG@customparshape}{}{\errmessage{failed to patch}}
\makeatother

\makeatletter
\def\thm@space@setup{%
	\thm@preskip=\parskip \thm@postskip=0pt
}
\makeatother

\usepackage{etoolbox}
\usepackage{tikz}
\usetikzlibrary{tikzmark}
\usetikzlibrary{calc}

\errorcontextlines\maxdimen

\newcommand{\ALGtikzmarkcolor}{black}
\newcommand{\ALGtikzmarkextraindent}{4pt}
\newcommand{\ALGtikzmarkverticaloffsetstart}{-.5ex}
\newcommand{\ALGtikzmarkverticaloffsetend}{-.5ex}
\makeatletter
\newcounter{ALG@tikzmark@tempcnta}

\newcommand\ALG@tikzmark@start{%
	\global\let\ALG@tikzmark@last\ALG@tikzmark@starttext%
	\expandafter\edef\csname ALG@tikzmark@\theALG@nested\endcsname{\theALG@tikzmark@tempcnta}%
	\tikzmark{ALG@tikzmark@start@\csname ALG@tikzmark@\theALG@nested\endcsname}%
	\addtocounter{ALG@tikzmark@tempcnta}{1}%
}

\def\ALG@tikzmark@starttext{start}
\newcommand\ALG@tikzmark@end{%
	\ifx\ALG@tikzmark@last\ALG@tikzmark@starttext
	\else
	\tikzmark{ALG@tikzmark@end@\csname ALG@tikzmark@\theALG@nested\endcsname}%
	\tikz[overlay,remember picture] \draw[\ALGtikzmarkcolor] let \p{S}=($(pic cs:ALG@tikzmark@start@\csname ALG@tikzmark@\theALG@nested\endcsname)+(\ALGtikzmarkextraindent,\ALGtikzmarkverticaloffsetstart)$), \p{E}=($(pic cs:ALG@tikzmark@end@\csname ALG@tikzmark@\theALG@nested\endcsname)+(\ALGtikzmarkextraindent,\ALGtikzmarkverticaloffsetend)$) in (\x{S},\y{S})--(\x{S},\y{E});%
	\fi
	\gdef\ALG@tikzmark@last{end}%
}

\apptocmd{\ALG@beginblock}{\ALG@tikzmark@start}{}{\errmessage{failed to patch}}
\pretocmd{\ALG@endblock}{\ALG@tikzmark@end}{}{\errmessage{failed to patch}}
\makeatother

%% file: intro.tex
\newcommand{\xf}[1]{f^{(#1)}}
\newcommand{\xG}[1]{G^{(#1)}}

\section{Introduction}
\label{sec:introduction}

In the \subcov problem, we are given a monotone, nonnegative
submodular function\footnotemark ${f: 2^\unvrs \rightarrow \Z_+}$, as well as a linear cost function $c$, 
and we wish to find the min-cost set $S \subseteq \unvrs$ such that $f(S) = f(\unvrs)$. \footnotetext{In the introduction we restrict to integer-valued functions for
simplicity; all results extend to general submodular functions with
suitable changes. See the technical sections for the full results
and nuanced details.} 
This is a classical NP-hard problem: e.g., when $f$ is a coverage
function we capture the \setcov problem. Moreover, the greedy
algorithm is known to be an $(1 +\ln \fmax)$-approximation, where
$\fmax$ is the maximum value of any single
element~\cite{wolsey1982analysis}. This bound is tight assuming
$\P \neq \NP$, even for the special case of \setcov~\cite{feige1998threshold, Dinur:2014:AAP:2591796.2591884}.

We consider this \subcov problem in a fully-dynamic setting,
where the notion of coverage changes over time. At each time, the
underlying submodular function changes from $\xf{t}$ to $\xf{t+1}$, and
the algorithm may have to change its solution from $S_t$ to $S_{t+1}$
to cover this new function. We do not want our solutions to change
wildly if the function changes by small amounts. The goal of this
work is to develop algorithms where this ``churn''
$|S_t \triangle S_{t+1}|$ is small (perhaps in an amortized sense),
while maintaining the requirement that each solution $S_t$ is a good
approximate solution to the function $\xf{t}$. The change
$|S_t \triangle S_{t+1}|$ is often called \emph{recourse} in the
literature.

This problem has been posed and answered in the special case of \setcov---for clarity, 
from now on this paper we consider the equivalent \hvc (a.k.a.\ the \hitset) problem. 
In this problem, hyperedges arrive to and depart from an active set over time, and we must maintain a small set of
vertices that hit all active hyperedges. We know an algorithm that maintains
an $O(\log n_t)$-approximation which has constant amortized
recourse~\cite{Gupta:2017:ODA:3055399.3055493}; here $n_t$ refers to
the number of active hyperedges at time $t$. In other words, the total
recourse---the total number of changes over $T$ edge arrivals and
departures---is only $O(T)$.
The algorithm and analysis are based on a delicate token-based
argument, which gives each hyperedge a constant number of tokens upon its
arrival, and moves these tokens in a careful way between edges to
account for the changes in the solution. What do we do in the more
general \subcov case, where there is no notion of sets any
more?

In this work we study the model where a submodular function is added or
removed from the \emph{active set} $\xG{t}$ at each timestep: this
defines the current submodular function
$\xf{t} := \sum_{g \in G_{t}} g$ as the sum of functions in this
active set. The algorithm must maintain a subset $S_t \sse \unvrs$ such
that $\xf{t}(S_t) = \xf{t}(\unvrs)$ with cost $c(S_t)$ being within a
small approximation factor of the optimal \subcov for
$\xf{t}$, such that the total recourse
$\sum_t | S_t \triangle S_{t+1} |$ remains small. To verify that this problem models 
dynamic \hvc, each arriving/departing edge $A_t$ should correspond to a
submodular function $g_t$ taking on value~$1$ for any set $S$ that hits
$A_t$, and value~zero otherwise (i.e., $g_t(S) = \mathbbm{1}[S \cap A_t \neq \emptyset]$).

\subsection{Our Results}
\label{sec:our-results}

Our main result is the following:

\begin{theorem}[Informal]
  \label{thm:intro-main_weighted_subc}
  There is a deterministic algorithm that maintains an
  $e^2\cdot(1 + \log \fmax)$- competitive solution to Submodular
  Cover in the fully-dynamic setting where functions arrive/depart over
  time. This algorithm has total recourse:
  \[ O\bigg(\sum_{t} g_t(\unvrs) \ln \Big(\frac{\cmax}{\cmin}\Big)
    \bigg),\] where $g_t(\unvrs)$ is the value of the function
  considered at time $t$, and $\cmax, \cmin$ are the maximum and
  minimum element costs.
\end{theorem}

Let us parse this result. Firstly, the approximation factor almost
matches Wolsey's result up to the multiplicative factor of $e^2$; this
is best possible in polynomial time unless $\P = \NP$ even in the
offline setting. Secondly, the amortized recourse bound should be thought of as
being logarithmic---indeed, specializing it to \hitset where each
$g_t(\unvrs) = 1$, we get a total recourse of
$O(T \log (\cmax/\cmin))$ over $T$ timesteps, or an amortized recourse
of $O(\log (\cmax/\cmin))$ per timestep. Hence this recourse bound is
weaker by a log-of-cost-spread factor, while generalizing to all monotone
submodular functions. Finally, since we are allowed to give richer
functions $g_t$ at each step, it is natural that the recourse
scales with $\sum_t g_t(\unvrs)$, which is the total ``volume'' of
these functions. In particular, this problem captures fully-dynamic $\hvc$ where at each round a batch of $k$ edges appears all at once (this is in contrast to the standard fully-dynamic model where hyperedges appear one at a time). In this case the algorithm may have to buy up to $k = g_t(\unvrs) / \fmin$ new vertices in general to maintain coverage.

We next show that for coverage functions (and hence for the \hvc
problem), a variation on the algorithm from
\Cref{thm:intro-main_weighted_subc} can remove the
log-of-cost-spread factor
in terms of recourse, at the cost of a slightly coarser competitive
ratio.
E.g., 
for \hvc the new competitive ratio corresponds to an $O(\log n)$
guarantee versus an $O(\log D_{\max})$ guarantee, where $D_{\max}$ being the largest degree of any vertex.

\begin{theorem}[Informal]
	\label{thm:intro-main_weighted_3incr}
	There is a deterministic algorithm that maintains an $O(\log f(\unvrs))$- competitive solution to \subcov in the fully-dynamic setting where functions arrive/depart over time, and each function is $3$-increasing in addition to monotone and submodular. Furthermore, this algorithm has total recourse:
	\[ O\bigg(\sum_t  g_t(\unvrs)\bigg). \]
\end{theorem}

Indeed, this result holds not just for coverage functions, but for the broader class of
\emph{$3$-increasing, monotone, submodular functions}
\cite{foldes2005submodularity}, which are the functions we have been
considering, with the additional property that
have positive discrete mixed third-derivatives. At a high level, these
are functions where the mutual coverage does not increase upon conditioning.

\subsection{Techniques and Overview}
\label{sec:techniques}

The most widely known algorithm for \subcov is the greedy algorithm;
this repeatedly adds to the solution an element maximizing the ratio
of its marginal coverage to its cost. It is natural to try to use the
greedy algorithm in our dynamic setting; the main issue is that
out-of-the-box, greedy may completely change its solution between time
steps. In their result on recourse-bounded \hvc,
\cite{Gupta:2017:ODA:3055399.3055493} showed how to effectively
imitate the greedy algorithm without sacrificing more than a small
amount of recourse. A barrier to making greedy algorithms dynamic is
their sequential nature, and hyperedge inserts/deletes can play havoc
with this. So they give a local-search algorithm that skirts this
sequential nature, while simultaneously retaining the key properties
of the greedy algorithm. Unfortunately, their analysis hinges on
delicately assigning edges to vertices. In the more
general \subcov setting, there are no
edges to track, and this approach breaks down.

Our first insight is to return to the sequential nature of the greedy
algorithm. Our algorithm maintains an ordering $\pi$ on the elements
of $\unvrs$, which induces a solution to the \subcov problem:
we define the output of the algorithm
as the prefix of elements that have non-zero marginal value. We
maintain this permutation dynamically via local updates, and argue
that only a small amount of recourse is necessary to ensure the
solution is competitive. To bound the competitive ratio, we imagine
that the permutation corresponds to the order in which an auxiliary
offline algorithm selects elements, i.e. a \emph{stack trace}. We show
that our local updates maintain the invariant that this is the stack
trace of an approximate greedy algorithm for the currently active set of
functions. We hope that this general framework of doing local search
\textit{on the stack trace of an auxiliary algorithm} finds uses in
other online/dynamic algorithms.

Our main technical contribution is to give a potential function to
argue that our algorithm needs bounded recourse. The potential
measures the \emph{(generalized) entropy} of the \emph{coverage vector} of
the permutation $\pi$. This coverage vector is indexed by elements of
the universe $\unvrs$, and the value of each coordinate is the
marginal coverage (according to permutation $\pi$) of the
corresponding element. Entropy is often used as a potential function
(notably, in recent developments in online algorithms) but in a
qualitatively different way. In many if not all cases, these
algorithms follow the maximum-entropy principle and seek \textit{high
  entropy} solutions subject to feasibility constraints; the potential
function then tracks the convergence to this goal. On the other hand,
in our setting the cost function is the support size of the coverage
vector, and minimizing this roughly corresponds to \textit{low
  entropy}. We show entropy decreases sufficiently quickly with every
change performed during the local search, and increases by a
controlled amount with insertion/deletion of each function $g$, thus
proving our recourse bounds.

We find our choice of entropy to be interesting for several
reasons. For one, we use a suitably chosen \emph{Tsallis
  entropy}~\cite{tsallis1988possible}, which is a general parametrized
family of entropies, instead of the classical Shannon entropy. The
latter also yields recourse bounds for our problem, but they are
substantially weaker (see \cref{sec:shannonentropy}). Tsallis entropy
has appeared in several recent algorithmic areas, for example as a
regularizer in bandit settings~\cite{seldin2014one}, and as an
approximation to Shannon entropy in streaming
contexts~\cite{harvey2008sketching}. Secondly, it is well known that
for certain problems, minimizing an $\ell_1$-objective is an effective
proxy for minimizing sparsity~\cite{candes2006stable}. In our
dynamics, the $\ell_1$ mass stays constant since total coverage does
not change when elements are
reordered. However, entropy is a good proxy for sparsity, in the sense 
that it decreases monotonically (and quickly!) as our algorithm moves within 
the $\ell_1$ level set towards sparse vectors.

In \cref{sec:unweightedsubmod}, we study the unit cost case to lay out
our framework and highlight the main ideas. We show a $\log \fmax /
\fmin$ competitive algorithm for fully-dynamic \subcov with
$O(\sum_t g_t(\unvrs) / \fmin)$ total recourse. Then in \cref{sec:weightedsubmod}, we turn to general cost functions. We again show a $\log (\fmax / \fmin)$ competitive algorithm, this time with $\log (\cmax / \cmin) \cdot \sum_t g_t (\unvrs) / \fmin$ recourse. The algorithm template is the same as before, but with a suitably generalized potential function and analysis. Here we also require a careful choice of the Tsallis entropy parameter, near (but not quite at) the limit point at which Tsallis entropy becomes Shannon entropy. 

In \cref{sec:3incr}, we show how to remove the $\log (\cmax / \cmin)$
dependence and achieve optimal recourse for a structured family of
monotone, submodular functions: the class of \emph{$3$-increasing (monotone, submodular) set functions}
\cite{foldes2005submodularity}. These are set functions, all of whose
discrete mixed third derivatives are nonnegative everywhere. Submodular
functions in this class include \textit{measure coverage functions},
which generalize set coverage functions, as well as entropies for
distributions with positive interaction information (see, e.g., \cite{jakulin2005machine} for a discussion). Since this class includes \setcov, this recovers the optimal $O(1)$ recourse bound of \cite{Gupta:2017:ODA:3055399.3055493}. For this result we use a more interesting generalization of the potential function from \cref{sec:unweightedsubmod} that reweights the coverage of elements in the permutation non-uniformly as a function of their mutual coverage with other elements of the permutation.

In \cref{sec:rjuntas}, we show how to get improved randomized algorithms when the functions $g_t$ are assumed to be $r$-juntas. This is in analogy to approximation algorithms for \setcov under the frequency parametrization. In \cref{sec:combiner}, we also show how to run an online ``combiner'' algorithm that gets the best of all worlds, with a competitive ratio of $O(\min(r, \log \fmax / \fmin))$.  Finally, in \cref{sec:furtherapps}, we demonstrate the generality of our framework by using it to recover known recourse bounded algorithms for fully-dynamic \mst and \mstt. These achieve $O(1)$ competitive ratios, and recourse bounds of $O(\log D)$ where $D$ is the aspect ratio of the metric. Our proofs here are particularly simple and concise.

\subsection{Related Work}
\label{sec:related-work}

\textbf{Submodular Cover.}
While we introduced the problem for integer-valued functions, all
results can be extended to real-valued settings by adding a dependence
on $\fmin$, the smallest marginal value.
Wolsey~\cite{wolsey1982analysis} showed that the greedy algorithm, repeatedly selecting the element maximizing marginal coverage to cost ratio, gives a $1 + \ln (\fmax / \fmin)$ approximation for \subcov; this guarantee is tight unless $\P = \NP$ even for the special case of \setcov \cite{feige1998threshold,Dinur:2014:AAP:2591796.2591884}. Fujito~\cite{fujito2000approximation}
gave a dual greedy algorithm that generalizes the
$F$-approximation for \setcov~\cite{doi:10.1137/0211045} where $F$ is
the maximum-frequency of an element. 

\subcov has been used in many applications to resource allocation and placement
problems, by showing that the coverage function is monotone and submodular,
and then applying Wolsey's greedy algorithm as a black box. 
We automatically port these applications to the fully-dynamic setting where coverage
requirements change with time. 
E.g., in selecting influential nodes to disseminate information in social networks \cite{Goyal2013, loukides2016limiting, tong2017positive, IZUMI20102773}, exploration for robotics planning problems \cite{krause2008robust, jorgensen2017risk, beinhofer2013robust}, placing sensors 
\cite{WU201553, rahimian2015detection, zheng2017trading, 7504484}, and
other physical resource allocation objectives \cite{6996018, 7798894,
	tzoumas2016minimal}. The networking community has been particularly interested in \subcov
recently, since \subcov models network function placement tasks 
\cite{Andreev:2009:SSL:1644015.1644031, Lee:2013:FPC:2523616.2525960,
	kortsarz2015approximating, LUKOVSZKI2018159, chen2018virtual}. E.g., 
\cite{LUKOVSZKI2018159} want to place middleboxes in a
network incrementally, and point out that avoiding closing extant boxes is a huge boon in practice.

The definition of \emph{$m$-increasing functions} is due to Foldes and
Hammer~\cite{foldes2005submodularity}. Bach~\cite{bach} gave a characterization of the class of \textit{measure coverage functions} (which we define later) in terms of its iterated discrete derivatives. This class generalizes coverage functions, and is contained in the class of $3$-increasing functions. 	\cite{iyer2020submodular} give several additional examples of $3$-increasing functions. Several papers \cite{iyer2020submodular, 8619396, wang2015accelerated, wang2013weapon} have given algorithms specifically for $3$-increasing submodular function optimization.

\textbf{Online and Dynamic Algorithms.}
There is a still budding series of work on recourse-bounded algorithms. Besides \cite{Gupta:2017:ODA:3055399.3055493} which is most directly related to our work, researchers have studied the Steiner Tree problem \cite{imase1991dynamic, gupta2014online, gu2016power, lkacki2015power}, clustering \cite{guo2020power, cohen2019fully}, matching \cite{grove1995online,chaudhuri2009online, bosek2014online}, and scheduling \cite{phillips1993online, westbrook2000load, andrews1999improved, sanders2009online, skutella2010robust, epstein2014robust, gupta2014maintaining}.

A rich parallel line of work has studied how to minimize update time
for problems in the dynamic or fully-dynamic setting. In
\cite{Gupta:2017:ODA:3055399.3055493}, the authors give an $O(\log n)$
competitive and $O(F \log n)$ update time algorithm for fully-dynamic \hvc. An ongoing
program of research for the frequency parametrization of \hvc
\cite{bhattacharya2018deterministic, bhattacharya2017deterministic,
  bhattacharya2019deterministically, abboud2019dynamic,
  bhattacharya2019new, bhattacharya2020improved} has so far culminated
in an ${F(1+\eps)}$ competitive algorithm with $\poly(F, \log \cmax / \cmin,
1/\eps)$ update time (where $F$ is the frequency).

In recent work, \cite{onlinesubmod} studied the problem of maintaining
a feasible solution to \subcov in a related online model. That setting
is an insertion-only irrevocable analog of this work, where functions
$g_t$ may never leave the active set. Our results can be seen as an
extension/improvement when recourse is allowed: not only can our
algorithm handle the fully-dynamic case with insertions and deletions,
but we improve the competitive ratio from $O(\log n \cdot \log \fmax /
\fmin)$ to $O(\log \fmax / \fmin)$, which is best possible even in the
offline case.

Our work is related to work on convex body chasing (e.g.,
\cite{argue2020chasing,sellke2020chasing}) in spirit but not in
techniques. For an online/dynamic covering problems, the set of
feasible fractional solutions within distance $\alpha$ of the optimal
solution at a given time step form a convex set: our goal is similarly
to ``chase'' these convex bodies, while limiting the total movement
traversed. The main difference is that we seek \emph{absolute bounds}
on the recourse, instead of recourse that is competitive with the
optimal chasing solution. (We can give such bounds because our
feasible regions are structured and not arbitrary convex bodies).

\subsection{Notation and Preliminaries}
\label{subsec:notation}

A set function $f: 2^{\unvrs} \rightarrow \mathbb{R}^+$
is \textit{submodular} if $f(A \cap B) + f(A \cup B) \leq f(A) + f(B)$
for any $A, B \subseteq \unvrs$.  It is \textit{monotone} if $f(A) \leq f(B)$
for all $A \subseteq B \subseteq \unvrs$. We assume access to a \emph{value
	oracle} for $f$ that computes $f(T)$ given $T \subseteq \unvrs$. The
\textit{contraction} of $f: 2^{\unvrs} \rightarrow \mathbb{R}^+$
onto $\unvrs\setminus T$ is defined as
$f_T(S) = f(S \mid T) := f(S \cup T) - f(T)$. 
If $f$ is submodular then $f_T$ is also submodular for any
$T \subseteq \unvrs$. We use the following notation.
\begin{align}
\fmax^{(t)} &:= \displaystyle \max \{f^{(t)}(j) \mid j \in N \},\\
\fmin^{(t)} &:= \displaystyle \min \{f^{(t)}(j \mid S) \mid j \in N, \ S \subseteq \unvrs, \ f^{(t)}(j \mid S) \neq 0\}. \\
\fmax &:= \displaystyle \max_t \fmax^{(t)}, \label{line:fmaxdef}\\
\fmin &:= \displaystyle \min_t \fmin^{(t)}. \label{line:fmindef}
\end{align}
Also we let $\cmax$ and $\cmin$ denote the largest and smallest costs of elements respectively. 

We will sometime use the simple and well known inequalities:
\begin{fact}
	\label{fact:averaging}
	Given positive numbers $a_1, \ldots, a_k$ and $b_1, \ldots, b_k$:
	\begin{align}
	\min_i \frac{a_i}{b_i} \leq \frac{\sum_i a_i}{\sum_i b_i} \leq \max_i \frac{a_i}{b_i}
	\end{align}
\end{fact}

Throughout this paper, we will use the convention that $1:k$ denotes that range of indices from $1$ to $k$.

\textbf{Mutual Coverage.} We will use the notion of mutual coverage defined in \cite{onlinesubmod}. Independently, \cite{iyer2020submodular} defined and studied the same quantity under the slightly different name \textit{submodular mutual information}.
 
\begin{definition}[Mutual Coverage]
	\label{def:mutualCoverage}
	The \emph{mutual coverage} and \emph{conditional mutual
		coverage} with respect to a set function
	$f: 2^{\unvrs} \rightarrow \mathbb{R}^+$ are
	defined as:
	\begin{align}
	\mutcov_f(A ; B) &:= f(A) + f(B) - f(A \cup B), \label{eq:mutcov} \\
	\mutcov_f(A ; B \mid C) &:= f_C(A) + f_C(B) - f_C(A \cup B). \label{eq:cond-mutcov}
	\end{align}
\end{definition}
We may think of $\mutcov_f{(A ; B \mid C)}$ intuitively as being the
amount of coverage $B$ ``takes away'' from the coverage of $A$ (or
vice-versa, since the definition is symmetric in $A$ and $B$), given
that $C$ was already chosen. This generalizes the notion of \emph{mutual information}
from information theory: if $\unvrs$ is a set of random variables, and 
$S\subseteq \unvrs$, and if $f(S)$ denotes  the joint entropy of the random
variables in the set $S$, then $\mutcov$ is the mutual information.

\begin{fact}[Chain Rule]
	\label{fact:chainRule}
	Mutual coverage respects the identity:
	\begin{align*}
	\mutcov_f(A ; B_1 \cup B_2 \mid C) = \mutcov_f(A; B_1 \mid C) + \mutcov_f(A; B_2 \mid C \cup B_1).
	\end{align*}	
\end{fact}
This neatly generalizes the chain rule for mutual information.


%% file: warmupI.tex
\section{Unit Cost Submodular Cover}
\label{sec:unweightedsubmod}

\subsection{The Algorithm}

We now present our first algorithm for unit-cost fully-dynamic \subcov. We will show the following:
\begin{theorem}
	\label{thm:main_subc}
	For any $\gamma > e$, there is a deterministic algorithm that maintains a $\gamma(\log \fmax^{(t)}/\fmin^{(t)} + 1)$-competitive solution to unweighted \subcov in the setting where functions arrive/depart over time. Furthermore, this algorithm has total recourse:
	\[ 2 \cdot \frac{e \ln \gamma}{\gamma - e \ln \gamma} \cdot \frac{\sum_t g_t(\unvrs)}{\fmin} = O\left(\frac{\sum_t g_t(\unvrs)}{\fmin}\right).\]
\end{theorem}

The algorithm and its analysis are particularly clean; we will build
on these in the following sections.
We begin by describing the algorithm. We maintain a permutation $\pi$
of the elements in $\unvrs$, and assign to each element its marginal
coverage given what precedes it in the permutation. We write this
marginal value 
assigned to element $\pi_i$ as 
\begin{gather}
  \mff_\pi(\pi_i) := f(\pi_i \mid \pi_{1:i-1}). \label{eq:marg-1}
\end{gather}

We consider two kinds of local search moves:
\begin{enumerate}
	\item \textbf{Swaps:} transform $\pi$ to $\pi'$ by moving an element at position $i$ to position $i-1$ on the condition that $\mff(\pi_i) \geq \mff(\pi_{i-1})$. In words, this is a \emph{bubble} operation (as in bubble-sort).
	\item  \textbf{$\gamma$-moves}: transform the permutation $\pi$ to $\pi'$ by moving an element $u$ from a position $q$ to some other position $p < q$ on the condition that for all $i \in \{p , \ldots, q-1\}$, 
	\[ \mff_{\pi'}(\pi'_p) \geq
	\gamma \cdot \mff_\pi (\pi_i). \]
	In words, when $u$ moves ahead in line, it ``steals'' coverage from other elements along the way; we
	require that the amount covered by $u$ \textit{after the
		move} (which is given by $\mff_{\pi'}(\pi'_p)$ since $u$ now sits at position
	$p$) to be at least a
	$\gamma$ factor larger than the coverage \textit{before the move} of any element that $u$ jumps over. (See~\Cref{fig:setcover}.)
\end{enumerate}

\begin{figure}
	\begin{mdframed}
	\tikzset{every picture/.style={line width=0.75pt}} 
	\begin{subfigure}{.5\textwidth}
		\centering
		\input{tikz/LocalMove1.tikz}
		\caption{Before Move}
	\end{subfigure}%
	\vline
	\begin{subfigure}{.5\textwidth}
		\centering
		\input{tikz/LocalMove2.tikz}
		\caption{After Move}
	\end{subfigure} 
	\caption{Illustration of a legal $\gamma$-move. Each rectangle
		represents the marginal coverage of an element of the
		permutation. The height of the item that moves must be at
		least $\gamma$ times the height of anyone it cuts in line.}
	\label{fig:setcover}
	\end{mdframed}
\end{figure}

The dynamic algorithm is the following. When a new function $g^{(t)}$
arrives or departs, update the coverages $\mff_\pi$ of all the
elements in the permutation. Subsequently, while there is a local
search move available, perform the move. Output the prefix of $\pi$ of
all elements with non-zero coverage. This is summarized in \cref{alg:dynamicSubC}, with a setting of $\gamma > e$. 

\begin{algorithm}[ht]
	\caption{\textsc{FullyDynamicSubmodularCover}}
	\label{alg:dynamicSubC}
	\begin{algorithmic}[1]
		\State $\pi \leftarrow$ arbitrary initial permutation of elements $\unvrs$.
		\For{$t = 1, 2, \ldots, T$}
		\State $t^{th}$ function $g_t$ arrives/departs.
		\While{there exists a legal $\gamma$-move or a swap for $\pi$}
		\State Perform the move, and update $\pi$.
		\EndWhile
		\State Output the collection of $\pi_i$ such that $\mff_\pi(\pi_i) > 0$.
		\EndFor 
	\end{algorithmic}
\end{algorithm}

\subsection{Bounding the Cost}
Let us start by arguing that if the algorithm terminates, it must produce a competitive solution.
\begin{lemma}
	\label{lem:sc_stablegood}
	Suppose no $\gamma$-moves are possible, then for every index $i$ such that $\mff_\pi(\pi_i) > 0$, and for every index $i' > i$, the following holds. Let $\pi'$ be the permutation where $\pi_{i'}$ is moved to position $i$. Then
	\begin{align}
		\mff_\pi(\pi_i) > \frac{\mff_{\pi'}(\pi_{i'})}{\gamma}
		\label{eq:gamma_greedy}
	\end{align}
\end{lemma}
\begin{proof}
	Suppose there are elements $\pi_i$ and $\pi_{i'}$ such that condition \eqref{eq:gamma_greedy} does not hold, i.e.
       $\mff_{\pi'}(\pi_j) \geq \gamma \cdot \mff_{\pi}(\pi_i)$.
        Since by assumption there are no swaps available, the permutation
        $\pi$ is in non-increasing order of $\mff_\pi(\pi_i)$
        values, so for all indices $j > i$ it also holds that $\mff_\pi (\pi_{i'}) \geq \gamma\cdot \mff_\pi(\pi_j)$. Hence moving
        $i'$ from its current position to position $i$ is a legal
        $\gamma$-move, which is a contradiction.
\end{proof}

\begin{corollary}
		\label{cor:stacktrace}
		The output at every time step is $\gamma \cdot (\log \fmax^{(t)} / \fmin^{(t)} + 1)$-competitive.
\end{corollary}

\begin{proof}
	\cref{lem:sc_stablegood} implies that the solution is equivalent to greedily selecting an element whose marginal coverage is within a factor of $1/
	\gamma$ of the largest marginal coverage currently available. Given this, the standard analyses of the greedy algorithm for \subcov \cite{wolsey1982analysis} imply that the solution is $\gamma \cdot (\log \fmax^{(t)} / \fmin^{(t)} + 1)$ competitive.
\end{proof}

\subsection{Bounding the Recourse}

\label{sec:uwtdsc-recourse}

We move on to showing that the algorithm must terminate with $O(g(\unvrs) / \fmin)$ amortized recourse. For this analysis, we define a potential function parametrized by a number $\alpha \in (0,1)$ to be fixed later:
\[\Phi_\alpha(f, \pi) := \sum_{i \in N} \left(\mff_\pi (\pi_i) \right)^\alpha.\]
As noted in the introduction, up to scaling and shifting, this is the Tsallis entropy with parameter $\alpha$.
We show several properties of this potential:

\begin{mdframed}
	\textbf{Properties of $\Phi_\alpha$:}
	
	\begin{enumerate}[label=\textbf{(\Roman*)}]
		\item $\Phi_\alpha$ increases by at most $g_t(\unvrs) \cdot (\fmin)^{\alpha-1}$ with the addition of function $g_t$ to the active set. \label{subc_rom1}
		\item $\Phi_\alpha$ does not increase with deletion of functions from the system. \label{subc_rom2}
		\item $\Phi_\alpha$ does not increase during swaps. \label{subc_rom3}
		\item For an appropriate range of $\gamma$, the potential $\Phi_\alpha$ decreases by at least $(\gamma / (e \ln \gamma) - 1)\cdot (\fmin)^\alpha = \Omega((\fmin)^\alpha)$ with every $\gamma$-move. \label{subc_rom4}
	\end{enumerate}
\end{mdframed}

\begin{lemma}
	\label{lem:uwtdsc_props}
	If $\alpha = (\ln \gamma)^{-1}$, then $\Phi_\alpha$ respects properties \labelcref{subc_rom1,subc_rom2,subc_rom3,subc_rom4}.
\end{lemma}

\begin{proof}
	For brevity, define $h(z) := z^\alpha$. Since $\alpha \in (0,1)$, this function is concave and non-decreasing.
	
	We start with property \ref{subc_rom1}. When a function $g$ is added to the system, for some set of $i \in [n]$, it increases $k_i := \mff_\pi(\pi(i))$ by some amount $\Delta_i$. Observe that $\sum_i \Delta_i = g(\unvrs)$. By the concavity of $h$:
	\[\sum_i h(k_i + \Delta_i) - \sum_i h(k_i) \leq \sum_i h(\Delta_i) \leq \frac{\sum_{i} \Delta_i}{\fmin} \cdot h(\fmin) = g(\unvrs) \cdot (\fmin)^{\alpha-1}.\]
	
	Property \ref{subc_rom2} follows since $h$ is non-decreasing. 
	
	For property \ref{subc_rom3}, we wish to show that if $u$ immediately precedes $v$ in $\pi$ and $\mff_\pi(u) \leq \mff_\pi(v)$, then swapping $u$ and $v$ does not increase the potential. Let $\widehat \pi$ denote the permutation after the swap. Note that $\mff_{\widehat \pi}(u) \leq \mff_\pi(v)$ and $\mff_{\widehat \pi}(v) \geq \mff_\pi(v)$,
	since $u$ may only have lost some amount of coverage to $v$. Suppose this amount is $k$, i.e. $k = {\mff_\pi(u) - \mff_{\widehat \pi}(u) = \mff_{\widehat \pi}(v) - \mff_\pi(v)}$. Then:
	\begin{align*}
		\Phi_\alpha(f, \widehat \pi) - \Phi_\alpha(f, \pi) &= h(\mff_{\widehat \pi}(u)) + h(\mff_{\widehat \pi}(v)) - h(\mff_{\pi}(u))- h(\mff_{\pi}(v))\\
		&= h(\mff_\pi(u) - k) + h(\mff_\pi(v) + k)- h(\mff_{\pi}(u))- h(\mff_{\pi}(v))
	\end{align*}
	which is non-positive due to the concavity of $h$. 
	
	It remains to prove property \ref{subc_rom4}. Suppose we perform
	a $\gamma$-move on a permutation $\pi$. Let $u$ be the element
	moving to some position $p$ from some position $q > p$, and let $\pi'$ denote the permutation after the move. For convenience, also define:
	\begin{align*}
		v_i &:= \mff_\pi (\pi_i), \tag{the original coverage of the $i^{th}$ set}\\
		a_i &:= \mutcov_f(\pi_i; u \mid \pi_{1:i-1}) = \mff_\pi (\pi_i) - \mff_{\pi'}(\pi_i). \tag{the loss in coverage of the $i^{th}$ set}
	\end{align*}
	Note that for all $i \not \in \{p, \ldots, q\}$, we necessarily have $a_i = 0$. Also note that $\mff_{\pi'}(S) = \sum_i a_i$, and by the definition of a $\gamma$-move, for any $j$ we have $\sum_i a_i \geq \gamma \cdot v_j$. Then the change in potential is:
	\begin{align}
		\Phi_\alpha(f, \pi') - \Phi_\alpha(f, \pi) &= \Big(\sum_i a_i\Big)^\alpha + \sum_i (v_i - a_i)^\alpha - \sum_i v_i^\alpha \notag \\
		&\leq \Big(\sum_i a_i\Big)^\alpha - \sum_i a_i \cdot \alpha \cdot v_i^{\alpha - 1} \label{eq:sc_concave} \\
		&\leq \Big(\sum_i a_i\Big)^\alpha - \Big(\sum_i a_i \Big)^{\alpha} \cdot \alpha \cdot \gamma^{1 - \alpha} \label{eq:sc_gammamove} \\
		&\leq - \Big(\frac{\gamma}{e \ln \gamma} - 1\Big) (\fmin)^\alpha. \label{eq:sc_gammachoice}
	\end{align}
	Above, \eqref{eq:sc_concave} holds since $h$ is concave and thus $h(v_i - a_i) - h(v_i) \leq \nabla h(v_i)  \cdot (-a_i)$.
	Line \eqref{eq:sc_gammamove} holds by the definition of a
	$\gamma$-move, since $\sum_i a_i \geq \gamma v_j$ for every $j
	\in \{p,\ldots, q\}$. Finally, \eqref{eq:sc_gammachoice} comes from our choice of $\alpha = (\ln \gamma)^{-1}$ and the fact that $\sum_i a_i \geq \fmin$. \end{proof}

Finally, we return to proving the main theorem.
\begin{proof}[Proof of \cref{thm:main_subc}]
	By \cref{lem:sc_stablegood}, if \cref{alg:dynamicSubC} (using Definition \ref{eq:marg-1} for $\mff_\pi$) terminates then it is $\gamma\cdot (\log \fmax^{(t)}/\fmin^{(t)} + 1)$-competitive.
	
	By \labelcref{subc_rom1,subc_rom4,subc_rom2,subc_rom3}, the potential $\Phi_\alpha$ increases by at most $g_t(\unvrs) (\fmin)^{\alpha-1}$ for every function $g_t$ inserted to the active set, decreases by $(\fmin)^\alpha \cdot \left(\gamma/(e\ln \gamma) - 1\right)$ per $\gamma$-move, and otherwise does not increase. By inspection, $\Phi_\alpha \geq 0$. The number of elements $e$ with $\mff_\pi(e) > 0$ grows by $1$ only during $\gamma$-moves in which $\mff_\pi(e)$ was initially $0$. Otherwise, this number never grows. We account for elements leaving the solution by paying recourse $2$ upfront when they join the solution.
	
	Hence, the number of changes to the solution is at most:
	\[ 2 \cdot \frac{\sum_t g_t(\unvrs)}{(\fmin)^{1 - \alpha}} \cdot\frac{e \ln \gamma}{(\fmin)^\alpha(\gamma - e \ln \gamma)} = O\left(\frac{\sum_t g_t(\unvrs)}{\fmin}\right). \qedhere \]
\end{proof}

Our algorithm gives a tunable tradeoff between approximation ratio and
recourse depending on the choice of $\gamma$. Note that if we wish to
optimize the competitive ratio, setting $\gamma = e(1+\delta)$ gives a
recourse bound of
\[ \left[ \left(1 + \frac{\ln (1+\delta)}{\delta - \ln(1+\delta)} \right) \right] \frac{\sum_t g_t(\unvrs)}{\fmin}= O\left(\frac{1}{\delta^2} \right) \frac{\sum_t g_t(\unvrs)}{\fmin}\] 
as $\delta$ approaches $0$. For simplicity one can set $\gamma = e^2$ to get the bound in \cref{thm:main_subc}.


%% file: tikz/LocalMove1.tikz
\tikzset{every picture/.style={line width=0.75pt}} 

\begin{tikzpicture}[x=0.75pt,y=0.75pt,yscale=-0.6,xscale=0.8]

\draw   (100,50) -- (130,50) -- (130,260) -- (100,260) -- cycle ;
\draw   (140,60) -- (170,60) -- (170,260) -- (140,260) -- cycle ;
\draw  [fill={rgb, 255:red, 255; green, 0; blue, 0 }  ,fill opacity=0.6 ] (180,180) -- (210,180) -- (210,260) -- (180,260) -- cycle ;
\draw  [fill={rgb, 255:red, 0; green, 255; blue, 0 }  ,fill opacity=0.6 ] (220,200) -- (250,200) -- (250,260) -- (220,260) -- cycle ;
\draw  [fill={rgb, 255:red, 0; green, 0; blue, 255 }  ,fill opacity=0.6 ] (260,210) -- (290,210) -- (290,260) -- (260,260) -- cycle ;
\draw  [fill={rgb, 255:red, 255; green, 163; blue, 0 }  ,fill opacity=0.6 ] (300,220) -- (330,220) -- (330,260) -- (300,260) -- cycle ;
\draw  [fill={rgb, 255:red, 0; green, 0; blue, 0 }  ,fill opacity=0.6 ] (340,230) -- (370,230) -- (370,260) -- (340,260) -- cycle ;
\draw   (380,240) -- (410,240) -- (410,260) -- (380,260) -- cycle ;
\draw    (360,210) .. controls (338.61,104.53) and (228.61,115.88) .. (180.72,169.19) ;
\draw [shift={(180,170)}, rotate = 311.34000000000003] [color={rgb, 255:red, 0; green, 0; blue, 0 }  ][line width=0.75]    (10.93,-3.29) .. controls (6.95,-1.4) and (3.31,-0.3) .. (0,0) .. controls (3.31,0.3) and (6.95,1.4) .. (10.93,3.29)   ;

\draw (190,287) node [anchor=south west][inner sep=0.75pt]    {$a$};
\draw (230,287) node [anchor=south west][inner sep=0.75pt]    {$b$};
\draw (270,287) node [anchor=south west][inner sep=0.75pt]    {$c$};
\draw (309,287) node [anchor=south west][inner sep=0.75pt]    {$d$};
\draw (349,287) node [anchor=south west][inner sep=0.75pt]    {$e$};

\end{tikzpicture}

%% file: tikz/LocalMove2.tikz

\tikzset{
	pattern size/.store in=\mcSize, 
	pattern size = 5pt,
	pattern thickness/.store in=\mcThickness, 
	pattern thickness = 0.3pt,
	pattern radius/.store in=\mcRadius, 
	pattern radius = 1pt}
\makeatletter
\pgfutil@ifundefined{pgf@pattern@name@_3phq0nvqd}{
	\pgfdeclarepatternformonly[\mcThickness,\mcSize]{_3phq0nvqd}
	{\pgfqpoint{0pt}{0pt}}
	{\pgfpoint{\mcSize+\mcThickness}{\mcSize+\mcThickness}}
	{\pgfpoint{\mcSize}{\mcSize}}
	{
		\pgfsetcolor{\tikz@pattern@color}
		\pgfsetlinewidth{\mcThickness}
		\pgfpathmoveto{\pgfqpoint{0pt}{0pt}}
		\pgfpathlineto{\pgfpoint{\mcSize+\mcThickness}{\mcSize+\mcThickness}}
		\pgfusepath{stroke}
}}
\makeatother


\tikzset{
	pattern size/.store in=\mcSize, 
	pattern size = 5pt,
	pattern thickness/.store in=\mcThickness, 
	pattern thickness = 0.3pt,
	pattern radius/.store in=\mcRadius, 
	pattern radius = 1pt}
\makeatletter
\pgfutil@ifundefined{pgf@pattern@name@_0m76376h0}{
	\pgfdeclarepatternformonly[\mcThickness,\mcSize]{_0m76376h0}
	{\pgfqpoint{0pt}{0pt}}
	{\pgfpoint{\mcSize+\mcThickness}{\mcSize+\mcThickness}}
	{\pgfpoint{\mcSize}{\mcSize}}
	{
		\pgfsetcolor{\tikz@pattern@color}
		\pgfsetlinewidth{\mcThickness}
		\pgfpathmoveto{\pgfqpoint{0pt}{0pt}}
		\pgfpathlineto{\pgfpoint{\mcSize+\mcThickness}{\mcSize+\mcThickness}}
		\pgfusepath{stroke}
}}
\makeatother


\tikzset{
	pattern size/.store in=\mcSize, 
	pattern size = 5pt,
	pattern thickness/.store in=\mcThickness, 
	pattern thickness = 0.3pt,
	pattern radius/.store in=\mcRadius, 
	pattern radius = 1pt}
\makeatletter
\pgfutil@ifundefined{pgf@pattern@name@_9p0f14ay5}{
	\pgfdeclarepatternformonly[\mcThickness,\mcSize]{_9p0f14ay5}
	{\pgfqpoint{0pt}{0pt}}
	{\pgfpoint{\mcSize+\mcThickness}{\mcSize+\mcThickness}}
	{\pgfpoint{\mcSize}{\mcSize}}
	{
		\pgfsetcolor{\tikz@pattern@color}
		\pgfsetlinewidth{\mcThickness}
		\pgfpathmoveto{\pgfqpoint{0pt}{0pt}}
		\pgfpathlineto{\pgfpoint{\mcSize+\mcThickness}{\mcSize+\mcThickness}}
		\pgfusepath{stroke}
}}
\makeatother


\tikzset{
	pattern size/.store in=\mcSize, 
	pattern size = 5pt,
	pattern thickness/.store in=\mcThickness, 
	pattern thickness = 0.3pt,
	pattern radius/.store in=\mcRadius, 
	pattern radius = 1pt}
\makeatletter
\pgfutil@ifundefined{pgf@pattern@name@_5l0e3qnng}{
	\pgfdeclarepatternformonly[\mcThickness,\mcSize]{_5l0e3qnng}
	{\pgfqpoint{0pt}{0pt}}
	{\pgfpoint{\mcSize+\mcThickness}{\mcSize+\mcThickness}}
	{\pgfpoint{\mcSize}{\mcSize}}
	{
		\pgfsetcolor{\tikz@pattern@color}
		\pgfsetlinewidth{\mcThickness}
		\pgfpathmoveto{\pgfqpoint{0pt}{0pt}}
		\pgfpathlineto{\pgfpoint{\mcSize+\mcThickness}{\mcSize+\mcThickness}}
		\pgfusepath{stroke}
}}
\makeatother
\tikzset{every picture/.style={line width=0.75pt}} 

\begin{tikzpicture}[x=0.75pt,y=0.75pt,yscale=-0.6,xscale=0.8]

\draw   (140,80) -- (170,80) -- (170,290) -- (140,290) -- cycle ;
\draw   (180,90) -- (210,90) -- (210,290) -- (180,290) -- cycle ;
\draw  [fill={rgb, 255:red, 255; green, 0; blue, 0 }  ,fill opacity=0.6 ] (260,270) -- (290,270) -- (290,290) -- (260,290) -- cycle ;
\draw  [fill={rgb, 255:red, 0; green, 255; blue, 0 }  ,fill opacity=0.6 ] (300,260) -- (330,260) -- (330,290) -- (300,290) -- cycle ;
\draw  [fill={rgb, 255:red, 0; green, 0; blue, 255 }  ,fill opacity=0.6 ] (340,270) -- (370,270) -- (370,290) -- (340,290) -- cycle ;
\draw  [fill={rgb, 255:red, 255; green, 163; blue, 0 }  ,fill opacity=0.6 ] (380,280) -- (410,280) -- (410,290) -- (380,290) -- cycle ;
\draw  [fill={rgb, 255:red, 0; green, 0; blue, 0 }  ,fill opacity=0.6 ] (220,260) -- (250,260) -- (250,290) -- (220,290) -- cycle ;
\draw   (420,270) -- (450,270) -- (450,290) -- (420,290) -- cycle ;
\draw  [pattern=_3phq0nvqd,pattern size=6pt,pattern thickness=0.75pt,pattern radius=0pt, pattern color={rgb, 255:red, 0; green, 0; blue, 0}] (260,210) -- (290,210) -- (290,270) -- (260,270) -- cycle ;
\draw  [fill={rgb, 255:red, 255; green, 0; blue, 0 }  ,fill opacity=0.6 ] (220,200) -- (250,200) -- (250,260) -- (220,260) -- cycle ;
\draw  [pattern=_0m76376h0,pattern size=6pt,pattern thickness=0.75pt,pattern radius=0pt, pattern color={rgb, 255:red, 0; green, 0; blue, 0}] (300,230) -- (330,230) -- (330,260) -- (300,260) -- cycle ;
\draw  [fill={rgb, 255:red, 0; green, 255; blue, 0 }  ,fill opacity=0.6 ] (220,170) -- (250,170) -- (250,200) -- (220,200) -- cycle ;
\draw  [pattern=_9p0f14ay5,pattern size=6pt,pattern thickness=0.75pt,pattern radius=0pt, pattern color={rgb, 255:red, 0; green, 0; blue, 0}] (340,240) -- (370,240) -- (370,270) -- (340,270) -- cycle ;
\draw  [fill={rgb, 255:red, 0; green, 0; blue, 255 }  ,fill opacity=0.6 ] (220,140) -- (250,140) -- (250,170) -- (220,170) -- cycle ;
\draw  [pattern=_5l0e3qnng,pattern size=6pt,pattern thickness=0.75pt,pattern radius=0pt, pattern color={rgb, 255:red, 0; green, 0; blue, 0}] (380,250) -- (410,250) -- (410,280) -- (380,280) -- cycle ;
\draw  [fill={rgb, 255:red, 255; green, 163; blue, 0 }  ,fill opacity=0.6 ] (220,110) -- (250,110) -- (250,140) -- (220,140) -- cycle ;


\draw (230,317) node [anchor=south west][inner sep=0.75pt]    {$e$};
\draw (270,317) node [anchor=south west][inner sep=0.75pt]    {$a$};
\draw (310,317) node [anchor=south west][inner sep=0.75pt]    {$b$};
\draw (349,317) node [anchor=south west][inner sep=0.75pt]    {$c$};
\draw (389,317) node [anchor=south west][inner sep=0.75pt]    {$d$};

\end{tikzpicture}

%% file: generalcosts.tex
\section{General Cost Submodular Cover}

\label{sec:weightedsubmod}

\subsection{The Algorithm}

We now turn to the general costs case and show the main result of our paper:
\begin{theorem}
	\label{thm:main_weighted_subc}
	There is a deterministic algorithm that maintains an $e^2\cdot(\log \fmax^{(t)}/\fmin^{(t)} + 1)$- competitive solution to \subcov in the setting where functions arrive/depart over time. Furthermore, this algorithm has amortized recourse:
	\[ O\left(\frac{g(\mathcal{N})}{\fmin} \ln \left(\frac{\cmax}{\cmin}\right) \right)\]
	per function arrival/departure.
\end{theorem}

Given the last section, our description of the new algorithm is very simple. We reuse \cref{alg:dynamicSubC}, except we redefine:
\[\mff_\pi(\pi_i) := \frac{f(\pi_i \mid \pi_{1:i-1})}{c(\pi_i)}. \]
We will specify the last free parameter $\gamma$ shortly.
\subsection{Bounding the Cost}

To bound the competitive ratio, note that \cref{eq:gamma_greedy} did not use any particular properties of $\mff_\pi$, except that in the solution output by the algorithm, there are no $\gamma$-moves or swaps with respect to $\mff_\pi$ in permutation $\pi$. The analog of \cref{cor:stacktrace} is nearly identical:

\begin{corollary}
	\label{cor:stacktrace2}
	The output at every time step is $\gamma \cdot (\log \fmax^{(t)} / \fmin^{(t)} + 1)$-competitive.
\end{corollary}

\begin{proof}
		\cref{lem:sc_stablegood} implies that the solution is equivalent to greedily selecting an element whose marginal coverage/cost ratio is within a factor of $1/\gamma$ of the largest marginal coverage/cost ratio currently available. The standard analyses of the greedy algorithm for \subcov \cite{wolsey1982analysis} imply that the solution is $\gamma \cdot (\log \fmax^{(t)} / \fmin^{(t)} + 1)$ competitive.
\end{proof}

\subsection{Bounding the Recourse}

To make our analysis as modular as possible, we will write a general potential function, parametrized by a function $h: \R \rightarrow \R $: 
\[\Phi_h(f, \pi) := \sum_{i \in N} c(\pi_i) \cdot h\left( \mff_{\pi}(\pi_i) \right). \]

With foresight, we require several properties of $h$:

\begin{mdframed}
	\textbf{Properties of $h$:}
\begin{enumerate}[label=\textbf{(\roman*)}]
	\item $h$ is monotone and concave. \label{gencost_srom1}
	\item $h(0) = 0$. \label{gencost_srom2}
	\item $h$ satisfies $x \cdot h'(x / \gamma) \geq (1+\eps_\gamma) h(x)$. \label{gencost_srom3}
	\item $h$ satisfies $y \cdot h (x/y)$ is monotone in $y$. \label{gencost_srom4}
\end{enumerate}
\end{mdframed}

\medskip

To bound the recourse, our goal will again be to show the following properties of our potential function $\Phi_h$:
\begin{mdframed}
	\textbf{Properties of $\Phi_h$:}
\begin{enumerate}[label=\textbf{(\Roman*)}]
	\item $\Phi_h$ increases by at most 
	\[\frac{g_t(\mathcal{N})}{\fmin} \cdot \cmax \cdot h\left(\frac{\fmin}{\cmax}\right)\]
	with the addition of function $g_t$ to the active set. \label{gencost_rom1}
	\item $\Phi_h$ does not increase with deletion of functions from the system. \label{gencost_rom2}
	\item $\Phi_h$ does not increase during sorting. \label{gencost_rom3}
	\item For an appropriate range of $\gamma$, the potential $\Phi_h$ decreases by at least 
	\[\eps_\gamma\cdot \cmin \cdot h\left(\frac{\fmin}{\cmin}\right)\]
	with every $\gamma$-move. \label{gencost_rom4}
\end{enumerate}
\end{mdframed}

Together, the statements imply a total recourse bound of:
\[\frac{\sum_t g_t(\mathcal{N})}{\eps_\gamma \cdot \fmin} \cdot \frac{\cmax}{\cmin} \cdot \frac{h(\fmin / \cmax)}{h(\fmin / \cmin)} \]

\begin{lemma}
	\label{lem:gencost_hprops}
	If $h$ respects properties \labelcref{gencost_srom1,gencost_srom2,gencost_srom3,gencost_srom4} then $\Phi_h$ respects properties \labelcref{gencost_rom1,gencost_rom2,gencost_rom3,gencost_rom4}.
\end{lemma}

\begin{proof}
	We start with property \ref{gencost_rom1}.
	When a function $g_t$ is added to the system, for some set of $i \in [n]$, it increases $k_i := \mff_\pi(\pi(i))$ by some amount $\Delta_i$. Then the potential increase is:
	\begin{align}
		\Phi_h(f^{(t)}, \pi) - \Phi_h(f^{(t-1)}, \pi) &= \sum_{i \in [n]} c(\pi_i) \cdot h\left( \frac{k_i + \Delta_i}{c(\pi_i)}\right) - \sum_{i \in [n]} c(\pi_i) \cdot h\left( \frac{k_i}{c(\pi_i)}\right) \notag \\
		&\leq \sum_{i \in [n]} c(\pi_i) \cdot h\left( \frac{\Delta_i}{c(\pi_i)}\right) \label{eq:gencost_using1and2} \\
		&\leq \sum_{i \in [n]} \cmax \cdot h \left(\frac{\Delta_i}{\cmax}\right) \label{eq:gencost_using4} \\
		&\leq \frac{\sum_{i} \Delta_i}{\fmin} \cdot \cmax \cdot h\left(\frac{\fmin}{\cmax}\right) \label{eq:gencost_using1} \\
		&= \frac{g_t(\mathcal{N})}{\fmin} \cdot \cmax \cdot h\left(\frac{\fmin}{\cmax}\right). \notag
	\end{align}
	
	Above step \eqref{eq:gencost_using1and2} is by properties \labelcref{gencost_srom1,gencost_srom2}, step \eqref{eq:gencost_using4} is by property \labelcref{gencost_srom4} and step \eqref{eq:gencost_using1} is by property \labelcref{gencost_srom1}.

	Property \ref{gencost_rom2} follows since $h$ is non-decreasing. 
	
	The proof of property \ref{gencost_rom3} is similar to the one in \cref{lem:uwtdsc_props}. Suppose $u$ immediately precedes $v$ in $\pi$ but $\mff_\pi(u) \leq \mff_\pi(v)$, and let $\widehat \pi$ denote the permutation after the swap. We have that $\mff_{\widehat \pi}(u) \leq \mff_\pi(v)$ and $\mff_{\widehat \pi}(v) \geq \mff_\pi(v)$,	since $u$ may only have lost some amount of coverage to $v$. Suppose this amount is $k$, i.e. $k = {\mff_\pi(u) - \mff_{\widehat \pi}(u) = \mff_{\widehat \pi}(v) - \mff_\pi(v)}$. Then:
	\begin{align*}
		\Phi_h(f, \widehat \pi) - \Phi_h(f, \pi) &= c(u) \cdot h\left(\mff_{\widehat \pi}(u)\right) + c(v) \cdot h(\mff_{\widehat \pi}(v)) - c(u) \cdot h(\mff_{\pi}(u))- c(v) \cdot h(\mff_{\pi}(v)) \\
		&= c(u) \left( h\left(\mff_{\pi}(u) - \frac{k}{c(u)} \right) - h(\mff_{\pi}(u))  \right) \\
		& \ \ \ + c(v) \left(h\left(\mff_{\pi}(v) + \frac{k}{c(v)} \right) -  h(\mff_{\pi}(v)) \right) \\
		&\leq k \cdot (h'(\mff_\pi(v)) -  h'(\mff_\pi(u)))
	\end{align*}
	which is non-positive due to the concavity of $h$ and the fact that $\mff_\pi (v) \geq \mff_\pi(u)$. 
	
	Finally the proof of property \ref{gencost_rom4} is also similar to the version in the last section. Suppose we perform a $\gamma$-move on a permutation $\pi$. Let $u$ be the element
	moving to some position $p$ from some position $q > p$, and let $\pi'$ denote the permutation after the move. Then:
	\begin{align}
		\Phi_h(f, \pi') - \Phi_h(f, \pi) &= c(u) \cdot h\left(\frac{\sum_{i \in [n]} a_i}{c(u)}\right) + \sum_{i \in [n]} c(\pi_i) \cdot h \left(\frac{v_i - a_i}{c(\pi_i)} \right) - \sum_{i \in [n]} c(\pi_i) \cdot h \left(\frac{v_i}{c(\pi_i)} \right) \notag \\
		&\leq c(u) \cdot h\left(\frac{\sum_{i \in [n]} a_i}{c(u)}\right)  - \sum_{i \in [n]} a_i \cdot h'\left(\frac{v_i}{c(\pi_i)}\right) \label{eq:gencost_usingconcavity} \\
		&\leq c(u) \cdot h\left(\frac{\sum_{i \in [n]} a_i}{c(u)}\right)  - c(u) \sum_{i \in [n]} \frac{a_i}{c(u)} \cdot h'\left(\frac{\sum_{i \in [n]} a_i}{\gamma \cdot c(u)}\right) \label{eq:gencost_gammamove} \\
		&\leq c(u) \cdot h\left(\frac{\sum_{i \in [n]} a_i}{c(u)}\right)  - (1+\eps_\gamma) \cdot c(u) \cdot h\left(\frac{\sum_{i \in [n]} a_i}{c(u)}\right) \label{eq:gencost_using3} \\
		&\leq - \eps_\gamma \cdot c(u) \cdot h\left(\frac{\sum_{i \in [n]} a_i}{c(u)}\right) \notag \\
		&\leq - \eps_\gamma \cdot \cmin \cdot h\left(\frac{\fmin}{\cmin}\right). \label{eq:gencost_using4again}
	\end{align}
	Above step \eqref{eq:gencost_usingconcavity} uses the concavity of $h$. Step \eqref{eq:gencost_gammamove} uses the fact that moving $u$ was a legal $\gamma$-move and thus $\sum_{i \in [n]} a_i / c(u) \geq \gamma v_i / c(\pi_i)$. Step \eqref{eq:gencost_using3} follows from property \labelcref{gencost_srom3}. Finally, \eqref{eq:gencost_using4again} uses property \labelcref{gencost_srom4} again.
\end{proof}

With this setup, the proof of the main theorem boils down to identifying an appropriate function $h$, and a suitable choice of $\gamma$. 

\begin{proof}[Proof of \cref{thm:main_weighted_subc}]
	Recall the general recourse bound is
	\[\frac{\sum_t g_t(\mathcal{N})}{\eps_\gamma \cdot \fmin} \cdot \frac{\cmax}{\cmin} \cdot \frac{h(\fmin / \cmax)}{h(\fmin / \cmin)}.\]
	A good choice for $h$ is $h(x) = x^{1-\delta}/ (1-\delta)$ for $\delta = \left(\ln (\cmax / \cmin)+1\right)^{-1}$, and with $\gamma = e^2$. Properties \labelcref{gencost_srom1,gencost_srom2,gencost_srom3,gencost_srom4} are easy to verify.	
	To see that this implies the bound
	\[O\left(\frac{\sum_t g_t(\mathcal{N})}{\fmin} \ln \left(\frac{\cmax}{\cmin}\right) \right),\]
	note that $\gamma = e^2 \geq ((1+\delta)/(1-\delta))^{1/\delta}$, in which case $\eps_\gamma = \gamma^\delta(1-\delta) - 1 \geq \delta$. Furthermore, $(\cmax / \cmin)^{\delta} = O(1)$.
\end{proof}

%% file: specialinstances.tex
\section{Improved Bounds for $3$-Increasing Functions}

\label{sec:3incr}

In this section we show to remove the $\log (\cmax / \cmin) $ dependence from the recourse bound when $f$ is assumed to be from a structured class of set functions: the class of $3$-increasing functions. We start with some definitions.

\subsection{Higher Order Monotonicity of Boolean Functions}

Following the notation of \cite{foldes2005submodularity}, we define the \textit{derivative of a set function} $f$ as:
\[\frac{df}{dx} (S) = f(S \cup \{x\}) - f(S \backslash \{x\}).\]
We notate the $m$-th order derivative of $f$ with respect to the subset $A = \{i_1, \ldots i_m\}$ as $df/dA$, and this quantity has the following concise expression:
\[\frac{df}{dA}(S) = \sum_{B \subseteq A} (-1)^{|B|} f((S \cup A) \backslash B).\]

\begin{definition}[$m$-increasing \cite{foldes2005submodularity}]
	\label{def:HOmonotone}
	We say that a set function is \textit{$m$-increasing} if all its $m$-th order derivatives are nonnegative, and \textit{$m$-decreasing} if they are nonpositive. We denote by $\mathcal{D}_m^+$ and $\mathcal{D}_m^-$ the classes of $m$-increasing and $m$-decreasing functions respectively.
\end{definition}

Note that $\mathcal{D}_1^+$ is the class of monotone set functions, and $\mathcal{D}_2^-$ is the class of submodular set functions. When $f$ is the joint entropy set function, the $m$-th derivative above is also known as the \textit{interaction information}, which generalizes the usual mutual information for two sets of variables, to $m$ sets of variables.

We will soon show improved algorithms for the class $\mathcal{D}_3^+$, that is the class of $3$-increasing set functions. The following derivation gives some intuition for these functions:
\begin{align}
	\frac{df}{d\{x,y,z\}}(S) &= \sum_{B \subseteq \{x, y, z\}} (-1)^{|B|} f((S \cup \{x,y,z\}) \backslash B) \notag \\
	&= f(x \mid S) - f(x \mid S \cup \{y\}) - ( f(x \mid S \cup \{z\}) - f(x \mid S \cup \{y, z\}) ) \notag \\
	&= \mutcov_f(x,y \mid S) - \mutcov_f(x,y \mid S \cup \{z\}). \label{eq:submodofconditioning}
\end{align}
Thus a function $f$ in contained in $\mathcal{D}_3^+$ if and only if mutual coverage decreases after conditioning.

Bach~\cite[Section 6.3]{bach} shows that a nonnegative function is in $\mathcal{D}_{2m-1}^+ \cap \mathcal{D}_{2m}^-$ simultaneously for all $m\geq 1$ if and only if it is a \textit{measure coverage function}\footnote{We avoid the term \textit{set coverage function} used by \cite{bach}, since we already use this terminology to mean the special case of a counting measure defined by a finite set system, as in \setcov.}: each element $i \in \mathcal{N}$ is associated with some measurable set $S_i$ under a measure $\mu$, and $f(S_1, \ldots, S_t) = \mu\left(\bigcup_{i=1}^t S_i \right)$.

\subsection{The Algorithm}

\label{subsec:3inc_algo}

We now show our second main result:

\begin{theorem}
	\label{thm:main_weighted_3incr}
	There is a deterministic algorithm that maintains an $O(\log f(N)/\fmin)$-competitive solution to \subcov in the fully-dynamic setting where functions arrive/depart over time, and each function is $3$-increasing in addition to monotone and submodular. Furthermore, this algorithm has total recourse:
	\[ O\left(\frac{\sum_t  g_t(\mathcal{N})}{\fmin}\right). \]
\end{theorem}

We (almost) exactly reuse \cref{alg:dynamicSubC} from previous sections, only this time
we redefine $\mff_\pi$ more substantially. Given two permutations
$\alpha$ and $\beta$ on $\unvrs$, we define the \emph{$(i,j)$ mutual
  affinity of $(\alpha, \beta)$} (and its conditional variant) as
\begin{align*}
    \mfI_{\alpha, \beta}(\alpha_i , \beta_j) &:= \mutcov_{f}(\alpha_i, \beta_j \mid \alpha_{1:i-1} \cup \beta_{1:j-1}), \\
    \mfI_{\alpha, \beta}(\alpha_i , \beta_j \mid S) &:= \mutcov_{f}(\alpha_i, \beta_j \mid \alpha_{1:i-1} \cup \beta_{1:j-1} \cup S).
\end{align*} 
Recall that $\mutcov_f$ denotes the \nameref{def:mutualCoverage}. 
To give some insight into these definitions, observe that mutual affinity telescopes cleanly:
\begin{observation}
	\label{obs:chainRuleExplained}
	The chain rule implies that
	\begin{align*}
	\sum_{j \in [n]} \mfI_{\alpha, \beta}(\alpha_i, \beta_j) &= f(\alpha_i \mid \alpha_{1:i-1}), \\
	\sum_{i \in [n]} \mfI_{\alpha, \beta}(\alpha_i, \beta_j) &= f(\beta_j \mid \beta_{1:j-1}), \\
	\sum_{i, j \in [n]} \mfI_{\alpha, \beta}(\alpha_i, \beta_j) &= f(\unvrs).
	\end{align*}
\end{observation}

Let $\psi$ denote the ordering of $\unvrs$ in \emph{increasing order of cost}. Then:
\begin{gather}
\mff_\pi(\pi_i) := \sum_{j \in [n]} \frac{\mfI_{\pi, \psi}(\pi_i , \psi_j)}{c(\pi_i) \cdot c(\psi_j)}. \label{eq:3incrmff}
\end{gather}
The following observation gives some intuition for this
definition. 
\begin{observation}
	\label{lem:wsc_costratio}
	The expression $\mfI_{\alpha, \beta}(\alpha_i, \beta_j)$
        is nonzero only if (a)~element $\alpha_i$ precedes
        element $\beta_j$ in permutation $\alpha$, and also (b)~element $\beta_j$ precedes element $\alpha_i$ in permutation $\beta$.
\end{observation}

In light of these remarks, $\mff_\pi(\pi_i)$ decomposes the marginal coverage of $\pi_i$ into its mutual affinity with all the elements $\psi_j \in \unvrs$ that are simultaneously cheaper than $\pi_i$ and that follow $\pi_i$ in the permutation $\pi$, and weights each of these affinities by $(c(\pi_i) c(\psi_j))^{-1}$.

With this re-definition of $\mff_\pi$, \cref{alg:dynamicSubC} is fully specified (though it is still parametrized by $\gamma$). We make one final important change, which begins with the following assumption.\footnote{This change to \cref{alg:dynamicSubC} is not present in the original version of this paper, see the discussion in \cref{sec:version_notes}.}

\begin{assumption}
    \label{assum:singleton_set}
    For any time $t \in [T]$, there are dummy elements $d^t_1, \ldots, d^t_n$ such that $d^t_i$ has cost $c(\psi_i)$, for any subset $S \subseteq [n]$, we have $g_t( \bigcup_{i \in S} \{d^t_i\}) = g_t( \bigcup_{i \in S} \{\psi_i\})$, but $g_{t'}(d^t_i) = 0$ for all $t' \neq t$.
\end{assumption}

This assumption is without loss of generality, because any solution using these dummy elements can be converted to a cheaper one using only the original universe elements by replacing $d^t_i$ with $\psi_i$.

The final change to \cref{alg:dynamicSubC} is this: when function $g_t$ arrives online at time $t$, first move $d^t_1, \ldots, d^t_{n}$ to the front of $\pi$ (arranged in  increasing order of cost), but skipping any  $d^t_i$ for which $g_t(d^t_i\mid d^t_{1:i-1}) = 0$. After this, proceed with the local search as in \cref{alg:dynamicSubC}.

 We move to proving formal guarantees.

\subsection{Bounding the Cost}

\label{sec:wsc_boundingcost}

Since our definition of $\mff_\pi$ (and thus the behavior of the algorithm) has significantly changed, we need to reprove the competitive ratio guarantee. Our goal will be the following Lemma:
\begin{lemma}
	\label{lem:wsc_stablegood}
	If no swaps or $\gamma$-moves are possible, the solution is $\gamma^2 \cdot \log (f^{(t)}(\unvrs) / \fmin^{(t)})$-competitive.
\end{lemma}

For the remainder of \cref{sec:wsc_boundingcost}, we drop the superscript and refer to $f^{(t)}$ as simply $f$.

Define a level $\Ell_\ell$ be the collection of elements $u$ such that $\mff_\pi(u) \in [\gamma^\ell, \gamma^{\ell+1})$. Note that since no swaps were possible, permutation $\pi$ is sorted in decreasing order of $\mff_\pi$, and thus $\Ell_\ell$ forms a contiguous interval of indices in $\pi$. Our proof strategy will be to show that for any level $\ell$, the total cost of all elements in $\Ell_\ell$ is at most $\gamma^2\cdot c(\opt)$. Subsequently, we will argue that there are at most $O(\log f(\unvrs) / \fmin)$ non-trivial levels.

\begin{lemma}[Each Level is Inexpensive]
	\label{claim:lvlcost}
	If there are no swaps or $\gamma$-moves for $\pi$, then for any $\ell > 0$, the total cost of all elements in $\Ell_\ell$ is at most $\gamma^2\cdot c(\opt)$.
\end{lemma}

\begin{proof}
	Suppose some level $\Ell_\ell$ has cost $c(\Ell_\ell) > \gamma^2 \cdot c(\opt)$. We will argue that in this case there must be a legal $\gamma$-move available. To start, using \cref{fact:averaging} we observe that:
	\begin{align}
		\gamma^{\ell} \leq \min_{u \in \Ell_\ell} \mff_\pi(u) \ \defeq \ \min_{u \in \Ell_\ell} \frac{\sum_{j \in [n]} \mfI_{\pi, \psi}(u , \psi_j)/  c(\psi_j)}{c(u)} &\leq \frac{\sum_{u\in \Ell_\ell} \sum_{j \in [n]} \mfI_{\pi, \psi}(u , \psi_j)/  c(\psi_j)}{\sum_{u \in \Ell_\ell} c(u)}, \notag
		\intertext{which by rearranging means:}
		\sum_{u\in \Ell_\ell} \sum_{j \in [n]} \frac{\mfI_{\pi, \psi}(u , \psi_j)}{ c(\psi_j)} \geq \gamma^\ell \cdot c(\Ell_\ell)&> \gamma^{\ell+2} \cdot c(\opt). \label{line:solCoversWell}
	\end{align}
	Let $\pi_{<}$ be the set of items preceding $\Ell_\ell$ in $\pi$, and $\pi_{\geq}$ be the set of items in or succeeding $\Ell_\ell$. We also define $\opt_{\geq} := \opt \cap \pi_{\geq}$.

	First, we imagine moving all the elements of $\opt_{\geq}$,
        simultaneously and in order according to $\pi$, to positions
        just before $\Ell_\ell$. Let $\pi'$ be this new
        permutation. We first show that some $o \in \opt$ has a high value
        $\mff_{\pi'}(o)$ after this move. Then we show that this
        element $o$ also constitutes a potential $\gamma$-move in $\pi$, a contradiction.

	For the first claim, we start by using \cref{fact:averaging} again:
		\begin{align}
	\max_{o \in \opt_{\geq}} \mff_{\pi'}(o) &\geq \frac{\sum_{o\in \opt_{\geq}} \sum_{j \in [n]} \mfI_{\pi', \psi}(o , \psi_j \mid \pi_<)/  c(\psi_j)}{\sum_{o \in \opt_{\geq}} c(o)} \notag \\
	&= \frac{1}{c(\opt_{\geq})} \cdot \sum_{o \in \opt_{\geq}} \sum_{j \in [n]} \frac{\mfI_{\pi', \psi}(o, \psi_j \mid \pi_{<})}{c(\psi_j)} \notag \\
	&= \frac{1}{c(\opt_{\geq})} \cdot \sum_{j \in [n]} \frac{f_{\psi}(\psi_j \mid \pi_{<})}{c(\psi_j)} \label{line:removeOptByChain} \\
	&\geq \frac{1}{c(\opt_{\geq})} \cdot \sum_{u \in \Ell_\ell} \sum_{j \in [n]} \frac{\mfI_{\pi, \psi}(u, \psi_j \mid \pi_{<})}{c(\psi_j)} \label{line:insertSolByChain}\\
	\intertext{The lines \eqref{line:removeOptByChain} and
		\eqref{line:insertSolByChain} above follow by
		\cref{obs:chainRuleExplained} (note that \eqref{line:insertSolByChain}) is an
		inequality because $\Ell_\ell$ may be a strict subset of
		$\pi_{\geq}$). Since $u \in \Ell_\ell \subseteq \pi_{\geq}$, this is:}
	&= \frac{1}{c(\opt_{\geq})} \cdot \sum_{u \in \Ell_\ell} \sum_{j \in [n]} \frac{\mfI_{\pi, \psi}(u, \psi_j)}{c(\psi_j)} \label{line:hatmfI_eq_mfI} \\
	&> \gamma^{\ell+2}, \label{line:applyingSolCoversWell}
	\end{align}
	where line \eqref{line:applyingSolCoversWell} used the inequality
	\eqref{line:solCoversWell}. In summary, $\mff_{\pi'}(o)$ is a
	factor $\gamma$ bigger than any of the $\mff_{\pi}(u)$ values
	for $u \in \Ell_{\ell}$.
	
		\begin{figure}
			\begin{mdframed}
		\tikzset{every picture/.style={line width=0.75pt}} 
		\begin{subfigure}{.5\textwidth}
			\centering
			\scalebox{0.77}{
				\input{tikz/3incr_1.tikz}}
			\vspace{0.01in}
			\caption{Permutation $\pi$.}
		\end{subfigure}%
		\vline
		\begin{subfigure}{.5\textwidth}
			\centering
			\scalebox{0.77}{
				\input{tikz/3incr_2.tikz}}
			\vspace{0.01in}
			\caption{Permutation $\pi'$, with $\opt_{\geq}$ moved ahead of $\Ell_\ell$.}
		\end{subfigure} 
		\caption{Illustration of the proof of \cref{claim:lvlcost}.}
		\label{fig:3incr_competitive}
		\end{mdframed}
	\end{figure}

        However, since permutation $\pi'$ is obtained by moving many
        elements of $\pi$ and not a single $\gamma$-move, we need to
        show that moving this element $o$ alone gives a legal
        $\gamma$-move in $\pi$. In fact, observe that we have not yet
        used that $f \in \mathcal{D}_3^+$: we will use it now.
	
	\begin{claim}
		Let $\pi^o$ be the permutation derived from $\pi$ where a single element $o \in \opt_{\geq}$ is moved before $\Ell_\ell$. Then $\mff_{\pi^o}(o) \geq \mff_{\pi'}(o)$.
	\end{claim}

	\begin{proof}
		Start from $\pi'$ and move an element $o' \in \opt_{\geq}$ with $o' \neq o$ back to its original position in $\pi$. Call this permutation $\pi''$. Then:
		\begin{align*}
		\mff_{\pi''}(o) &= \sum_{j \in [n]} \frac{\mfI_{\pi'', \psi}(o, \psi_j)}{c(o) c(\psi_j)} \\
		&= \sum_{j \in [n]} \frac{\mfI_{\pi', \psi}(o, \psi_j) + (\mfI_{\pi'', \psi}(o, \psi_j) - \mfI_{\pi'', \psi}(o, \psi_j \mid \{o'\}))}{c(o) c(\psi_j)}\\	
		&\geq \sum_{j \in [n]} \frac{\mfI_{\pi', \psi}(o, \psi_j)}{c(o) c(\psi_j)},
		\end{align*}
		where the second equality used $\mfI_{\pi', \psi}(o, \psi_j) = \mfI_{\pi'', \psi}(o, \psi_j \mid \{o'\})$, and 
		where the inequality used equation \eqref{eq:submodofconditioning}, that $\mutcov_f(a,b) - \mutcov_f(a, b \mid c) \geq 0$ for $3$-increasing functions. Repeating this process inductively until all elements but $o$ have been returned to their original positions in $\pi$ yields the permutation $\pi^o$, which proves the claim.
	\end{proof}
	Now we can complete the proof of \cref{claim:lvlcost}. This means that if $\mff_{\pi'}(o) \geq \gamma^{\ell+2}$, then
        there is a legal $\gamma$-move (namely the one which moves $o$
        ahead of $\Ell_\ell$), because by assumption every $u \in
        \pi_{\geq}$ has $\mff_\pi(u) \leq
        \gamma^{\ell+1}$. This contradicts the assumption that none existed.
\end{proof}

Next, we argue that the cost of all elements with very
high or very low values of $\mff_\pi$ is small.
\begin{lemma}[Extreme Values Lemma]
	\label{claim:triviallvls}
	If there are no swaps or $\gamma$-moves for $\pi$, the following hold:
	\begin{enumerate}[label=\textbf{(\roman*)}]
		\item There are no elements $u$ such that  $0 < \mff_\pi(u) \leq \fmin / (\gamma \cdot (c(\opt))^2)$ \label{item:3incr_lem1}.
		\item The total cost of all elements $u$ such that $\mff_\pi(u) \geq (f(\unvrs))^2 / (\fmin \cdot (c(\opt))^2)$ is at most $\sqrt{\gamma} \cdot \opt$ \label{item:3incr_lem2}.
	\end{enumerate}
\end{lemma}

\begin{proof}[Proof of \cref{claim:triviallvls}]
	To prove item \labelcref{item:3incr_lem1}, we observe that if permutation $\pi$ has no local moves, then every element $u$ must have high enough $\mff_\pi(u)$ to prevent any elements that follow it from cutting it in line.
	\begin{claim}
		\label{lem:lbOnmff}
		If there are no swaps or $\gamma$-moves for $\pi$, then for every $\pi_i \in \unvrs$, and every $j\in [n]$ such that $\mfI_{\pi, \psi}(\pi_i, \psi_j)>0$ we have:
		\[\mff_\pi(\pi_i) > \frac{f(\psi_j \mid \pi_{1:i-1})}{\gamma \cdot (c(\psi_j))^2}.\]
	\end{claim}
	
	\begin{proof}
		If there are no swaps or $\gamma$ moves, then $\mff_{\pi'}(\psi_j) < \gamma \cdot \mff_{\pi_i}$ (where $\pi'$ is the permutation obtained from $\pi$ by moving the element $\psi_j$ to the position ahead of $u$). Expanding definitions:
		\[\mff_{\pi'}(\psi_j) \ \defeq \ \sum_{j' \in [n]} \frac{\mfI_{\pi', \psi}(\psi_j, \psi_{j'})}{c(\psi_j)c(\psi_{j'})} \geq \sum_{j' \in [n]} \frac{\mfI_{\pi', \psi}(\psi_j, \psi_{j'})}{(c(\psi_j))^2} = \frac{f(\psi_j \mid \pi_{1:i-1})}{(c(\psi_j))^2}.\]
		In the first inequality we used that $c(\psi_{j'}) \leq c(\psi_j)$ by \cref{lem:wsc_costratio}, in the second equality we used \cref{obs:chainRuleExplained}. Rearranging terms yields the claim.
	\end{proof}
	
	Now item \labelcref{item:3incr_lem1} follows by setting $\psi_j$ to be the cheapest element that succeeds $u$ in the permutation. Note that $c(\psi_j) \leq c(\opt)$.
	
	For item \labelcref{item:3incr_lem2}, let $S$ be the set of indices $i$ with $\mff_\pi(\pi_i) \geq (f(\unvrs))^2 / (\fmin \cdot (c(\opt))^2)$. Then by \cref{fact:averaging}:
	\begin{align}
	\frac{f(\unvrs)}{c(\opt) \sqrt{\fmin}} &\leq \min_{i \in S} \sqrt{\mff_{\pi}(\pi_i)} = \min_{i \in S} \frac{\mff_\pi(\pi_i)}{\sqrt{\mff_\pi(\pi_i)}} \\
	&\defeq \ \frac{1}{c(\pi_i)} \sum_{j \in [n]} \frac{\mfI_{\pi, \psi}(\pi_i, \psi_j)}{c(\psi_j) \sqrt{\mff_\pi(\pi_i)}} \\
	&\leq \frac{1}{c(S)} \cdot \sum_{i\in S} \sum_{j \in [n]} \frac{\mfI_{\pi, \psi}(\pi_i, \psi_j)}{c(\psi_j) \sqrt{\mff_\pi(\pi_i)}} \notag \\
	\intertext{Rearranging to bound the cost:}
	c(S) &\leq \frac{ c(\opt) \sqrt{\fmin} }{f(\unvrs)} \cdot \sum_{i \in S} \sum_{j \in [n]} \frac{\mfI_{\pi, \psi}(\pi_i, \psi_j)}{c(\psi_j) \sqrt{\mff_\pi(\pi_i)}} \notag \\
	&\leq \frac{\sqrt{\gamma} \cdot  c(\opt) \sqrt{\fmin}}{f(\unvrs)} \cdot \sum_{i \in S} \sum_{j \in [n]} \frac{\mfI_{\pi, \psi}(\pi_i, \psi_j)}{\sqrt{f(\psi_j \mid \pi_{1:i-1})}} \label{eq:wsc_lvlboundstability} \\
	&\leq \frac{\sqrt{\gamma} \cdot  c(\opt) \sqrt{\fmin}}{f(\unvrs)} \cdot \sum_{i \in S} \sum_{j \in [n]} \frac{\mfI_{\pi, \psi}(\pi_i, \psi_j)}{\sqrt{\fmin}} \notag\\
	&\leq \frac{\sqrt{\gamma} \cdot  c(\opt)}{f(\unvrs)} \cdot \sum_{i \in S} \sum_{j \in [n]} \mfI_{\pi, \psi}(\pi_i, \psi_j) \notag\\
	&= \sqrt{\gamma} \cdot c(\opt). \label{eq:wsc_lvlboundCR}
	\end{align}
	Above, step \eqref{eq:wsc_lvlboundstability} is due to \cref{lem:lbOnmff} again and noting that if the denominator is 0 then the numerator is also 0, and step \eqref{eq:wsc_lvlboundCR} is due to \cref{obs:chainRuleExplained}.
\end{proof}

Using these ingredients, we can now wrap up the proof of the competitive ratio.
\begin{proof}[Proof of \cref{lem:wsc_stablegood}]
	The number of levels $\ell$ such that $\gamma^\ell$ lies between $\fmin / (\gamma (c(\opt))^2)$ and $(f(\unvrs))^2 / (\fmin \cdot (c(\opt))^2)$ is $O(\log_{\gamma} (f(\unvrs) / \fmin)^2) = O(\log (f(\unvrs) / \fmin))$. By \cref{claim:lvlcost}, each level in this range has cost at most $O(c(\opt))$. By \cref{claim:triviallvls}, there are no elements with non-zero coverage in levels below this range, and the total cost of all elements above this range is $O(c(\opt))$. Thus the total cost of the solution is $O(\log (f(\unvrs) / \fmin))\cdot c(\opt)$. 
\end{proof}

\subsection{Bounding the Recourse}

We follow our recipe of using the modified Tsallis entropy as a potential (with a reminder that the underlying definitions of $\mff_\pi$ have changed) with $\alpha$ fixed to $1/2$:
\begin{align*}
	\Phi_{1/2}(\pi) &:= \sum_{i=1}^n c(\pi_i) \sqrt{\mff_\pi(\pi_i)}.
\end{align*}
It is worthwhile to interpret these quantities in the context of \hvc. In this case, $\psi_i$ and $\psi_j$ correspond to vertices. Let $\Gamma(v)$ (and $\Gamma(V)$) denote the edge-neighborhood of $v$ (and the union of the edge neighborhoods of vertices in the set $V$), then 
\[\mfI_{\pi, \psi}(\pi_i , \psi_j) = |(\Gamma(\pi_i) \cap \Gamma(\psi_j)) \backslash (\Gamma(\pi_{1:i-1} \cup \psi_{1:j-1}))|,\] 
and if we use $c(e)$ to denote the cost of the cheapest vertex hitting edge $e$, we can simplify
\[c(\pi_i) \cdot \sqrt{\mff_\pi(\pi_i)} := \sqrt{\sum_{e \in \Gamma(\pi_i) \backslash \Gamma(\pi_{1:i-1})} \frac{c(\pi_i)}{c(e)}}. \]
In words, we are reweighting each hyperedge by the ratio of costs between the current vertex covering it, and its cheapest possible vertex that could cover it. Intuitively, this means the potential will be high when many elements are in sets that are significantly more expensive than the cheapest sets they could lie in.

\begin{figure}
	\begin{mdframed}
	\tikzset{every picture/.style={line width=0.75pt}} 
	\centering
	\scalebox{0.9}{
	\input{tikz/SubmodVol.tikz}}
	\caption{Illustration of $\Phi_{1/2}$. Elements are arranged in order of $\pi$.}
	\end{mdframed}
\end{figure}

This time the properties of $\Phi_{1/2}$ we show are:

\begin{mdframed}
	\textbf{Properties of $\Phi_{1/2}$:}
	\begin{enumerate}[label=\textbf{(\Roman*)}]
		\item $\Phi_{1/2}$ increases by at most $g(\unvrs) / \sqrt{\fmin}$ with every addition of a function to the system.  \label{wsc_rom1}
		\item $\Phi_{1/2}$ does not increase with deletion of functions from the system.  \label{wsc_rom2}
		\item $\Phi_{1/2}$ does not increase during swaps. \label{wsc_rom3}
		\item If $\gamma > 4$, then $\Phi_{1/2}$ decreases by at least $\Omega(\sqrt{\fmin})$ with every $\gamma$-move. \label{wsc_rom4}
	\end{enumerate}
\end{mdframed}

Together, these will yield a recourse bound of $\sum_t g_t(\unvrs)/\fmin$.

\begin{lemma}
\label{lem:phihalf_increase}
	$\Phi_{1/2}$ satisfies property \ref{wsc_rom1}.
\end{lemma}

\begin{proof}
Consider a step in which the submodular function $f$ changes to $f'$ because of the addition of a function $g_t$. Recall that the algorithm first moves dummy elements $d^t_1, \ldots, d^t_{n}$ to the front of $\pi$, but skipping any  $d^t_i$ for which $g_t(d^t_i\mid d^t_{1:i-1}) = 0$. Let $\pi'$ be the permutation permutation after this initial step, and let $\pi$ be the permutation before.

The contribution to $\Phi(f', \pi')$ of non-dummy is exactly the total potential $\Phi(f, \pi)$ before the function $g_t$ arrived. Hence the total potential change is exactly the contribution of the  dummy elements $d^t_1, \ldots, d^t_{n}$:
\begin{align*}\Phi(f', \pi') - \Phi(f, \pi) 
&= \sum_{i=1}^n \sqrt{\sum_{j=1}^n \frac{c(d^t_i)}{c(\psi_j)} \cdot g_t(d^t_i \mid d^t_{1:i-1}) \cdot \mathbbm{1}(i = j)} = 
\sum_{i=1}^n \sqrt{g_t(d^t_i \mid d^t_{1:i-1})} 
\end{align*}
 Since $g_t(d^t_i\mid d^t_{1:i-1}) > 0$, by definition $g_t(d^t_i\mid d^t_{1:i-1}) \geq \fmin$.
 By concavity of the square root function, this potential change is largest in the case when each term $g_t(d^t_i \mid d^t_{1:i-1})$ is exactly $\fmin$, and there are $g_t(\unvrs)/\fmin$ terms. Thus the total potential increase is at most $g_t(\unvrs)/\sqrt{\fmin}$. 
\end{proof}

\begin{lemma}
	$\Phi_{1/2}$ satisfies property \ref{wsc_rom3}.
\end{lemma}

\begin{proof}
	Consider an index $i$ such that swapping $\pi_i$ and $\pi_{i+1}$ increases the potential. Then for some quantity $\delta := \sum_{j \in [n]} (\mfI_{\pi, \psi}(\pi_i, \psi_j) - \mfI_{\pi', \psi}(\pi_i, \psi_j)) / c(j) > 0$ (since $f \in \mathcal{D}_3^+$) we have:
	\begin{align}
		&0 < \Delta \Phi_{1/2} \notag \\
		&0 < c(\pi_{i+1}) \left(\sqrt{\mff_\pi(\pi_{i+1}) + \frac{\delta}{c(\pi_{i+1})}} - \sqrt{\mff_\pi(\pi_{i+1})} \right) + c(\pi_i) \left(\sqrt{\mff_\pi(\pi_i) - \frac{\delta}{c(\pi_{i})}} - \sqrt{\mff_\pi(\pi_i)} \right) \notag\\
		&0 < \frac{\delta}{2 \sqrt{\mff_\pi (\pi_{i+1})}} - \frac{\delta}{2 \sqrt{\mff_\pi (\pi_{i})}} \label{eq:swapconcave} \\
		&\Rightarrow \hspace{0.2in} \mff_\pi(\pi_{i+1}) \leq \mff_\pi(\pi_i). \notag
	\end{align}
	Above, \eqref{eq:swapconcave} holds since square root is a concave function and thus $\sqrt{a + b} - \sqrt{a} \leq b / (2\sqrt{a})$. This implies that the local move was not a legal swap.
\end{proof}

\begin{lemma}
	If $\gamma > 4$, then $\Phi_{1/2}$ satisfies property \ref{wsc_rom4}.
\end{lemma}

\begin{proof}
	Suppose the local move changes the permutation $\pi$ to $\pi'$ by moving element $u$ from position $q$ to $p$. For notational convenience, define the following quantities:
	\begin{align*}
		v_i &:= \mff_\pi (\pi_i), \\
		a_{ij} &:=\mfI_\pi(\pi_i, \psi_j) - \mfI_{\pi'}(\pi_i, \psi_j) \ \defeq \ \mathbbm{1}[p \leq i \leq q] \cdot \left(\mfI_\pi(\pi_i, \psi_j) - \mfI_{\pi}(\pi_i, \psi_j \mid \{u\})\right).
	\end{align*}
	Recall that by \eqref{eq:submodofconditioning}, the quantity $a_{ij}$ is (crucially) nonnegative when $f$ is $3$-increasing. Also, by expanding the definition of \nameref{def:mutualCoverage}, for all indices $i \in \{p, \ldots, q \}$ we can rewrite $a_{ij}$ as:
	\begin{align*}
	a_{ij} &= f(\psi_j \mid \pi_{1:i-1} \cup \psi_{1:j-1}) - f(\psi_j \mid \pi_{1:i} \cup \psi_{1:j-1}) \\
	& \ \ \ - [f(\psi_j \mid \pi_{1:i-1} \cup \psi_{1:j-1} \cup \{u\}) - f(\psi_j \mid \pi_{1:i} \cup \psi_{1:j-1} \cup \{u\})] \\
	&= \mutcov_f(u, \psi_j \mid \pi_{1:i-1} \cup \psi_{1:j-1}) - \mutcov_f(u, \psi_j \mid \pi_{1:i} \cup \psi_{1:j-1}).
	\intertext{Thus, by the \nameref{fact:chainRule} these terms telescopes such that:}
	\sum_{i \in \{p, \ldots, q \}} a_{ij} &= \mutcov_f(u, \psi_j \mid \pi_{1:p-1}, \psi_{1:j-1}) \ \defeq \ \mfI_{\pi', \psi}(u, \psi_j).
	\end{align*}
	
	 With this, we can bound the potential change after a local move as
	\begin{align}
		&\Phi_{1/2}(f, \pi') - \Phi_{1/2}(f, \pi) \notag \\
		&= c(u) \sqrt{\mff_{\pi'}(u)} + \sum_{i \in [n]} c(\pi_i) \sqrt{v_i - \sum_{j \in [n]} \frac{a_{ij}}{c(\pi_i) c(\psi_j)}} - \sum_{i \in [n]} c(\pi_i) \sqrt{v_i} \notag \\
		&\leq c(u) \sqrt{\mff_{\pi'}(u)} - c(\pi_i) \sum_{i \in [n]} \frac{1}{2 \sqrt{v_i}} \cdot \sum_{j \in [n]} \frac{a_{ij}}{c(\pi_i) c(\psi_j)} \label{eq:wsc_convavity} \\
		&\leq c(u) \sqrt{\mff_{\pi'}(u)} - \frac{\sqrt{\gamma}}{2} \cdot \frac{c(\pi_i)}{\sqrt{\mff_{\pi'}(u)}}\cdot \sum_{i \in [n]} \sum_{j \in [n]} \frac{a_{ij}}{c(\pi_i) c(\psi_j)} \label{eq:wsc_gammamove} \\
		&= c(u) \sqrt{\mff_{\pi'}(u)} - \frac{\sqrt{\gamma}}{2} \cdot \frac{c(u)}{\sqrt{\mff_{\pi'}(u)}}\cdot \sum_{j \in [n]} \frac{\sum_{i \in [n]} a_{ij}}{c(u) c(\psi_j)} \notag \\
		\intertext{Above, step \eqref{eq:wsc_convavity} holds since square root is a concave function and thus $\sqrt{a + b} - \sqrt{a} \leq b / (2\sqrt{a})$. The next step \eqref{eq:wsc_gammamove} is due to the definition of $\gamma$-moves which ensure that $v_i \leq \mff_{\pi'}(u) / \gamma$ (note that this step also makes use of the nonnegativity of $a_{ij}$). Using the telescoping property of the $a_{ij}$, we can continue:}
		&= c(u) \sqrt{\mff_{\pi'}(u)} - \frac{\sqrt{\gamma}}{2} \cdot \frac{c(u)}{\sqrt{\mff_{\pi'}(u)}}\cdot \sum_{j \in [n]} \frac{\mfI_{\pi', \psi} (u, \psi_j)}{c(u) c(\psi_j)} \notag \\
		& \defeq - \left(\frac{\sqrt{\gamma}}{2} -1\right) \sqrt{\sum_{j \in [n]} \frac{c(u) \cdot \mfI_{\pi', \psi}(u, \psi_j)}{c(\psi_j)}} \notag \\
		&\leq - \left(\frac{\sqrt{\gamma}}{2} -1\right) \sqrt{\sum_{j \in [n]} \mfI_{\pi', \psi}(u, \psi_j)} \label{eq:wsc_psiorder} \\
		&\leq - \left(\frac{\sqrt{\gamma}}{2} -1\right) \sqrt{\fmin}. \label{eq:wsc_CR}
	\end{align}
	Step \eqref{eq:wsc_psiorder} comes from \cref{lem:wsc_costratio}, and finally step \eqref{eq:wsc_CR} follows by \cref{obs:chainRuleExplained} and the fact that $\fmin$ is a lower bound on marginal coverage for any element with nonzero marginal coverage.
\end{proof}

We wrap up with the proof of the main theorem.
\begin{proof}[Proof of \cref{thm:main_weighted_3incr}]	
	Set $\gamma = 5 > 4$. By \cref{lem:wsc_stablegood}, if \cref{alg:dynamicSubC} (using Definition \ref{eq:3incrmff} for $\mff_\pi$) terminates then it is $O(\log f(N) / \fmin)$-competitive.
	
	By \labelcref{wsc_rom1,wsc_rom4,wsc_rom2,wsc_rom3}, the potential $\Phi_{1/2}$ increases by at most $g_t(\unvrs / \sqrt{\fmin}$ for every function $g_t$ inserted to the active set, decreases by $\sqrt{\fmin} \cdot \left(\sqrt{\gamma}/2 - 1\right)$ per $\gamma$-move, and otherwise does not increase. By inspection, $\Phi_\alpha \geq 0$. The number of elements $e$ with $\mff_\pi(e) > 0$ grows by $1$ only during $\gamma$-moves in which $\mff_\pi(e)$ was initially $0$. Otherwise, this number never grows. We account for elements leaving the solution by paying recourse $2$ upfront when they join the solution.
	
	Hence, the number of changes to the solution is at most:
	\[ 2 \cdot \frac{\sum_t g_t(\unvrs)}{\sqrt{\fmin}} \cdot\frac{2}{\sqrt{\fmin}(\sqrt{\gamma} - 2)} = O\left(\frac{\sum_t g_t(\unvrs)}{\fmin}\right). \qedhere \]
\end{proof}


%% file: tikz/3incr_1.tikz
\tikzset{every picture/.style={line width=0.75pt}} 

\begin{tikzpicture}[x=0.75pt,y=0.75pt,yscale=-1,xscale=1]

\draw  [dash pattern={on 0.84pt off 2.51pt}]  (100,126) -- (470,125) ;
\draw [line width=6]    (240,125) -- (410,125) ;
\draw  [draw opacity=0][fill={rgb, 255:red, 255; green, 0; blue, 0 }  ,fill opacity=1 ] (280,125) .. controls (280,122.24) and (282.24,120) .. (285,120) .. controls (287.76,120) and (290,122.24) .. (290,125) .. controls (290,127.76) and (287.76,130) .. (285,130) .. controls (282.24,130) and (280,127.76) .. (280,125) -- cycle ;
\draw  [draw opacity=0][fill={rgb, 255:red, 255; green, 0; blue, 0 }  ,fill opacity=1 ] (316,125) .. controls (316,122.24) and (318.24,120) .. (321,120) .. controls (323.76,120) and (326,122.24) .. (326,125) .. controls (326,127.76) and (323.76,130) .. (321,130) .. controls (318.24,130) and (316,127.76) .. (316,125) -- cycle ;
\draw  [draw opacity=0][fill={rgb, 255:red, 255; green, 0; blue, 0 }  ,fill opacity=1 ] (216,125) .. controls (216,122.24) and (218.24,120) .. (221,120) .. controls (223.76,120) and (226,122.24) .. (226,125) .. controls (226,127.76) and (223.76,130) .. (221,130) .. controls (218.24,130) and (216,127.76) .. (216,125) -- cycle ;
\draw  [draw opacity=0][fill={rgb, 255:red, 255; green, 0; blue, 0 }  ,fill opacity=1 ] (426,125) .. controls (426,122.24) and (428.24,120) .. (431,120) .. controls (433.76,120) and (436,122.24) .. (436,125) .. controls (436,127.76) and (433.76,130) .. (431,130) .. controls (428.24,130) and (426,127.76) .. (426,125) -- cycle ;
\draw  [draw opacity=0][fill={rgb, 255:red, 255; green, 0; blue, 0 }  ,fill opacity=1 ] (126,125) .. controls (126,122.24) and (128.24,120) .. (131,120) .. controls (133.76,120) and (136,122.24) .. (136,125) .. controls (136,127.76) and (133.76,130) .. (131,130) .. controls (128.24,130) and (126,127.76) .. (126,125) -- cycle ;
\draw  [draw opacity=0][fill={rgb, 255:red, 255; green, 0; blue, 0 }  ,fill opacity=1 ] (456,125) .. controls (456,122.24) and (458.24,120) .. (461,120) .. controls (463.76,120) and (466,122.24) .. (466,125) .. controls (466,127.76) and (463.76,130) .. (461,130) .. controls (458.24,130) and (456,127.76) .. (456,125) -- cycle ;
\draw [color={rgb, 255:red, 255; green, 0; blue, 0 }  ,draw opacity=1 ]   (285,120) .. controls (278.66,87.82) and (252.84,98.43) .. (240.89,118.44) ;
\draw [shift={(240,120)}, rotate = 298.71] [color={rgb, 255:red, 255; green, 0; blue, 0 }  ,draw opacity=1 ][line width=0.75]    (10.93,-3.29) .. controls (6.95,-1.4) and (3.31,-0.3) .. (0,0) .. controls (3.31,0.3) and (6.95,1.4) .. (10.93,3.29)   ;
\draw [color={rgb, 255:red, 255; green, 0; blue, 0 }  ,draw opacity=1 ]   (321,120) .. controls (305.73,90.45) and (275.91,90.98) .. (247.78,118.71) ;
\draw [shift={(246.5,120)}, rotate = 314.5] [color={rgb, 255:red, 255; green, 0; blue, 0 }  ,draw opacity=1 ][line width=0.75]    (10.93,-3.29) .. controls (6.95,-1.4) and (3.31,-0.3) .. (0,0) .. controls (3.31,0.3) and (6.95,1.4) .. (10.93,3.29)   ;
\draw [color={rgb, 255:red, 255; green, 0; blue, 0 }  ,draw opacity=1 ]   (431,120) .. controls (353.68,72.72) and (286.54,90.45) .. (255.88,118.7) ;
\draw [shift={(254.5,120)}, rotate = 315.97] [color={rgb, 255:red, 255; green, 0; blue, 0 }  ,draw opacity=1 ][line width=0.75]    (10.93,-3.29) .. controls (6.95,-1.4) and (3.31,-0.3) .. (0,0) .. controls (3.31,0.3) and (6.95,1.4) .. (10.93,3.29)   ;
\draw [color={rgb, 255:red, 255; green, 0; blue, 0 }  ,draw opacity=1 ]   (461,120) .. controls (397.47,70.75) and (296.1,91.36) .. (264.88,118.74) ;
\draw [shift={(263.5,120)}, rotate = 316.49] [color={rgb, 255:red, 255; green, 0; blue, 0 }  ,draw opacity=1 ][line width=0.75]    (10.93,-3.29) .. controls (6.95,-1.4) and (3.31,-0.3) .. (0,0) .. controls (3.31,0.3) and (6.95,1.4) .. (10.93,3.29)   ;
\draw   (409.5,81) .. controls (409.5,76.33) and (407.17,74) .. (402.5,74) -- (341.33,74) .. controls (334.66,74) and (331.33,71.67) .. (331.33,67) .. controls (331.33,71.67) and (328,74) .. (321.33,74)(324.33,74) -- (247.5,74) .. controls (242.83,74) and (240.5,76.33) .. (240.5,81) ;

\draw (321,40.4) node [anchor=north west][inner sep=0.75pt]    {$\mathcal{L}_{\ell }$};
\draw (125,132.4) node [anchor=north west][inner sep=0.75pt]  [color={rgb, 255:red, 255; green, 0; blue, 0 }  ,opacity=1 ]  {$o_{1}$};
\draw (215,132.4) node [anchor=north west][inner sep=0.75pt]  [color={rgb, 255:red, 255; green, 0; blue, 0 }  ,opacity=1 ]  {$o_{2}$};
\draw (277,132.4) node [anchor=north west][inner sep=0.75pt]  [color={rgb, 255:red, 255; green, 0; blue, 0 }  ,opacity=1 ]  {$o_{3}$};
\draw (313,132.4) node [anchor=north west][inner sep=0.75pt]  [color={rgb, 255:red, 255; green, 0; blue, 0 }  ,opacity=1 ]  {$o_{4}$};
\draw (424,132.4) node [anchor=north west][inner sep=0.75pt]  [color={rgb, 255:red, 255; green, 0; blue, 0 }  ,opacity=1 ]  {$o_{5}$};
\draw (455,132.4) node [anchor=north west][inner sep=0.75pt]  [color={rgb, 255:red, 255; green, 0; blue, 0 }  ,opacity=1 ]  {$o_{6}$};

\phantom{
\draw (124,132.4) node [anchor=north west][inner sep=0.75pt]  [color={rgb, 255:red, 255; green, 0; blue, 0 }  ,opacity=1 ]  {$o_{1}$};
\draw (214,132.4) node [anchor=north west][inner sep=0.75pt]  [color={rgb, 255:red, 255; green, 0; blue, 0 }  ,opacity=1 ]  {$o_{2}$};
\draw (223,164.4) node [anchor=north west][inner sep=0.75pt]  [color={rgb, 255:red, 255; green, 0; blue, 0 }  ,opacity=1 ]  {$o_{3}$};
\draw (243,164.4) node [anchor=north west][inner sep=0.75pt]  [color={rgb, 255:red, 255; green, 0; blue, 0 }  ,opacity=1 ]  {$o_{4}$};
\draw (263,164.4) node [anchor=north west][inner sep=0.75pt]  [color={rgb, 255:red, 255; green, 0; blue, 0 }  ,opacity=1 ]  {$o_{5}$};
\draw (283,164.4) node [anchor=north west][inner sep=0.75pt]  [color={rgb, 255:red, 255; green, 0; blue, 0 }  ,opacity=1 ]  {$o_{6}$};}

\end{tikzpicture}

%% file: tikz/3incr_2.tikz
\tikzset{every picture/.style={line width=0.75pt}} 

\begin{tikzpicture}[x=0.75pt,y=0.75pt,yscale=-1,xscale=1]

\draw  [dash pattern={on 0.84pt off 2.51pt}]  (100,126) -- (470,125) ;
\draw [line width=6]    (280,125) -- (450,125) ;
\draw  [draw opacity=0][fill={rgb, 255:red, 255; green, 0; blue, 0 }  ,fill opacity=1 ] (240,125) .. controls (240,122.24) and (242.24,120) .. (245,120) .. controls (247.76,120) and (250,122.24) .. (250,125) .. controls (250,127.76) and (247.76,130) .. (245,130) .. controls (242.24,130) and (240,127.76) .. (240,125) -- cycle ;
\draw  [draw opacity=0][fill={rgb, 255:red, 255; green, 0; blue, 0 }  ,fill opacity=1 ] (250,125) .. controls (250,122.24) and (252.24,120) .. (255,120) .. controls (257.76,120) and (260,122.24) .. (260,125) .. controls (260,127.76) and (257.76,130) .. (255,130) .. controls (252.24,130) and (250,127.76) .. (250,125) -- cycle ;
\draw  [draw opacity=0][fill={rgb, 255:red, 255; green, 0; blue, 0 }  ,fill opacity=1 ] (216,125) .. controls (216,122.24) and (218.24,120) .. (221,120) .. controls (223.76,120) and (226,122.24) .. (226,125) .. controls (226,127.76) and (223.76,130) .. (221,130) .. controls (218.24,130) and (216,127.76) .. (216,125) -- cycle ;
\draw  [draw opacity=0][fill={rgb, 255:red, 255; green, 0; blue, 0 }  ,fill opacity=1 ] (260,125) .. controls (260,122.24) and (262.24,120) .. (265,120) .. controls (267.76,120) and (270,122.24) .. (270,125) .. controls (270,127.76) and (267.76,130) .. (265,130) .. controls (262.24,130) and (260,127.76) .. (260,125) -- cycle ;
\draw  [draw opacity=0][fill={rgb, 255:red, 255; green, 0; blue, 0 }  ,fill opacity=1 ] (126,125) .. controls (126,122.24) and (128.24,120) .. (131,120) .. controls (133.76,120) and (136,122.24) .. (136,125) .. controls (136,127.76) and (133.76,130) .. (131,130) .. controls (128.24,130) and (126,127.76) .. (126,125) -- cycle ;
\draw  [draw opacity=0][fill={rgb, 255:red, 255; green, 0; blue, 0 }  ,fill opacity=1 ] (270,125) .. controls (270,122.24) and (272.24,120) .. (275,120) .. controls (277.76,120) and (280,122.24) .. (280,125) .. controls (280,127.76) and (277.76,130) .. (275,130) .. controls (272.24,130) and (270,127.76) .. (270,125) -- cycle ;
\draw   (450.5,107) .. controls (450.5,102.33) and (448.17,100) .. (443.5,100) -- (382.33,100) .. controls (375.66,100) and (372.33,97.67) .. (372.33,93) .. controls (372.33,97.67) and (369,100) .. (362.33,100)(365.33,100) -- (288.5,100) .. controls (283.83,100) and (281.5,102.33) .. (281.5,107) ;
\draw [color={rgb, 255:red, 255; green, 0; blue, 0 }  ,draw opacity=1 ]   (233.5,164) -- (240.93,138.92) ;
\draw [shift={(241.5,137)}, rotate = 466.5] [color={rgb, 255:red, 255; green, 0; blue, 0 }  ,draw opacity=1 ][line width=0.75]    (10.93,-3.29) .. controls (6.95,-1.4) and (3.31,-0.3) .. (0,0) .. controls (3.31,0.3) and (6.95,1.4) .. (10.93,3.29)   ;
\draw [color={rgb, 255:red, 255; green, 0; blue, 0 }  ,draw opacity=1 ]   (251.5,163) -- (254.26,139.99) ;
\draw [shift={(254.5,138)}, rotate = 456.84] [color={rgb, 255:red, 255; green, 0; blue, 0 }  ,draw opacity=1 ][line width=0.75]    (10.93,-3.29) .. controls (6.95,-1.4) and (3.31,-0.3) .. (0,0) .. controls (3.31,0.3) and (6.95,1.4) .. (10.93,3.29)   ;
\draw [color={rgb, 255:red, 255; green, 0; blue, 0 }  ,draw opacity=1 ]   (269.5,164) -- (265.82,140.97) ;
\draw [shift={(265.5,139)}, rotate = 440.91] [color={rgb, 255:red, 255; green, 0; blue, 0 }  ,draw opacity=1 ][line width=0.75]    (10.93,-3.29) .. controls (6.95,-1.4) and (3.31,-0.3) .. (0,0) .. controls (3.31,0.3) and (6.95,1.4) .. (10.93,3.29)   ;
\draw [color={rgb, 255:red, 255; green, 0; blue, 0 }  ,draw opacity=1 ]   (287.5,166) -- (277.28,141.84) ;
\draw [shift={(276.5,140)}, rotate = 427.07] [color={rgb, 255:red, 255; green, 0; blue, 0 }  ,draw opacity=1 ][line width=0.75]    (10.93,-3.29) .. controls (6.95,-1.4) and (3.31,-0.3) .. (0,0) .. controls (3.31,0.3) and (6.95,1.4) .. (10.93,3.29)   ;

\draw (362,66.4) node [anchor=north west][inner sep=0.75pt]    {$\mathcal{L}_{\ell }$};
\draw (124,132.4) node [anchor=north west][inner sep=0.75pt]  [color={rgb, 255:red, 255; green, 0; blue, 0 }  ,opacity=1 ]  {$o_{1}$};
\draw (214,132.4) node [anchor=north west][inner sep=0.75pt]  [color={rgb, 255:red, 255; green, 0; blue, 0 }  ,opacity=1 ]  {$o_{2}$};
\draw (223,164.4) node [anchor=north west][inner sep=0.75pt]  [color={rgb, 255:red, 255; green, 0; blue, 0 }  ,opacity=1 ]  {$o_{3}$};
\draw (243,164.4) node [anchor=north west][inner sep=0.75pt]  [color={rgb, 255:red, 255; green, 0; blue, 0 }  ,opacity=1 ]  {$o_{4}$};
\draw (263,164.4) node [anchor=north west][inner sep=0.75pt]  [color={rgb, 255:red, 255; green, 0; blue, 0 }  ,opacity=1 ]  {$o_{5}$};
\draw (283,164.4) node [anchor=north west][inner sep=0.75pt]  [color={rgb, 255:red, 255; green, 0; blue, 0 }  ,opacity=1 ]  {$o_{6}$};

\phantom{
\draw (321,40.4) node [anchor=north west][inner sep=0.75pt]    {$\mathcal{L}_{\ell }$};
\draw   (409.5,81) .. controls (409.5,76.33) and (407.17,74) .. (402.5,74) -- (341.33,74) .. controls (334.66,74) and (331.33,71.67) .. (331.33,67) .. controls (331.33,71.67) and (328,74) .. (321.33,74)(324.33,74) -- (247.5,74) .. controls (242.83,74) and (240.5,76.33) .. (240.5,81) ;
}

\end{tikzpicture}

%% file: tikz/SubmodVol.tikz
	

 
\tikzset{
pattern size/.store in=\mcSize, 
pattern size = 5pt,
pattern thickness/.store in=\mcThickness, 
pattern thickness = 0.3pt,
pattern radius/.store in=\mcRadius, 
pattern radius = 1pt}
\makeatletter
\pgfutil@ifundefined{pgf@pattern@name@_dr16lhqv1}{
\pgfdeclarepatternformonly[\mcThickness,\mcSize]{_dr16lhqv1}
{\pgfqpoint{0pt}{-\mcThickness}}
{\pgfpoint{\mcSize}{\mcSize}}
{\pgfpoint{\mcSize}{\mcSize}}
{
\pgfsetcolor{\tikz@pattern@color}
\pgfsetlinewidth{\mcThickness}
\pgfpathmoveto{\pgfqpoint{0pt}{\mcSize}}
\pgfpathlineto{\pgfpoint{\mcSize+\mcThickness}{-\mcThickness}}
\pgfusepath{stroke}
}}
\makeatother

 
\tikzset{
pattern size/.store in=\mcSize, 
pattern size = 5pt,
pattern thickness/.store in=\mcThickness, 
pattern thickness = 0.3pt,
pattern radius/.store in=\mcRadius, 
pattern radius = 1pt}
\makeatletter
\pgfutil@ifundefined{pgf@pattern@name@_3oy3oiylz}{
\pgfdeclarepatternformonly[\mcThickness,\mcSize]{_3oy3oiylz}
{\pgfqpoint{0pt}{-\mcThickness}}
{\pgfpoint{\mcSize}{\mcSize}}
{\pgfpoint{\mcSize}{\mcSize}}
{
\pgfsetcolor{\tikz@pattern@color}
\pgfsetlinewidth{\mcThickness}
\pgfpathmoveto{\pgfqpoint{0pt}{\mcSize}}
\pgfpathlineto{\pgfpoint{\mcSize+\mcThickness}{-\mcThickness}}
\pgfusepath{stroke}
}}
\makeatother

 
\tikzset{
pattern size/.store in=\mcSize, 
pattern size = 5pt,
pattern thickness/.store in=\mcThickness, 
pattern thickness = 0.3pt,
pattern radius/.store in=\mcRadius, 
pattern radius = 1pt}
\makeatletter
\pgfutil@ifundefined{pgf@pattern@name@_tg6xg38o2}{
\pgfdeclarepatternformonly[\mcThickness,\mcSize]{_tg6xg38o2}
{\pgfqpoint{0pt}{-\mcThickness}}
{\pgfpoint{\mcSize}{\mcSize}}
{\pgfpoint{\mcSize}{\mcSize}}
{
\pgfsetcolor{\tikz@pattern@color}
\pgfsetlinewidth{\mcThickness}
\pgfpathmoveto{\pgfqpoint{0pt}{\mcSize}}
\pgfpathlineto{\pgfpoint{\mcSize+\mcThickness}{-\mcThickness}}
\pgfusepath{stroke}
}}
\makeatother

 
\tikzset{
pattern size/.store in=\mcSize, 
pattern size = 5pt,
pattern thickness/.store in=\mcThickness, 
pattern thickness = 0.3pt,
pattern radius/.store in=\mcRadius, 
pattern radius = 1pt}
\makeatletter
\pgfutil@ifundefined{pgf@pattern@name@_lbm9x7hkc}{
\pgfdeclarepatternformonly[\mcThickness,\mcSize]{_lbm9x7hkc}
{\pgfqpoint{0pt}{-\mcThickness}}
{\pgfpoint{\mcSize}{\mcSize}}
{\pgfpoint{\mcSize}{\mcSize}}
{
\pgfsetcolor{\tikz@pattern@color}
\pgfsetlinewidth{\mcThickness}
\pgfpathmoveto{\pgfqpoint{0pt}{\mcSize}}
\pgfpathlineto{\pgfpoint{\mcSize+\mcThickness}{-\mcThickness}}
\pgfusepath{stroke}
}}
\makeatother
\tikzset{every picture/.style={line width=0.75pt}} 

\begin{tikzpicture}[x=0.75pt,y=0.75pt,yscale=-0.6,xscale=1]

\draw   (100,50) -- (130,50) -- (130,260) -- (100,260) -- cycle ;
\draw  [fill={rgb, 255:red, 255; green, 163; blue, 0 }  ,fill opacity=0.6 ] (180,80) -- (210,80) -- (210,120) -- (180,120) -- cycle ;
\draw   (140,70) -- (170,70) -- (170,260) -- (140,260) -- cycle ;
\draw   (220,170) -- (250,170) -- (250,260) -- (220,260) -- cycle ;
\draw  [fill={rgb, 255:red, 0; green, 220; blue, 0 }  ,fill opacity=0.6 ] (180,120) -- (210,120) -- (210,180) -- (180,180) -- cycle ;
\draw  [fill={rgb, 255:red, 0; green, 0; blue, 255 }  ,fill opacity=0.6 ] (180,180) -- (210,180) -- (210,210) -- (180,210) -- cycle ;
\draw  [fill={rgb, 255:red, 255; green, 0; blue, 0 }  ,fill opacity=0.6 ] (180,210) -- (210,210) -- (210,260) -- (180,260) -- cycle ;
\draw  [pattern=_dr16lhqv1,pattern size=3pt,pattern thickness=0.75pt,pattern radius=0pt, pattern color={rgb, 255:red, 255; green, 0; blue, 0}] (260,190) -- (290,190) -- (290,260) -- (260,260) -- cycle ;
\draw  [pattern=_3oy3oiylz,pattern size=3pt,pattern thickness=0.75pt,pattern radius=0pt, pattern color={rgb, 255:red, 0; green, 0; blue, 255}] (420,230) -- (450,230) -- (450,260) -- (420,260) -- cycle ;
\draw  [pattern=_tg6xg38o2,pattern size=3pt,pattern thickness=0.75pt,pattern radius=0pt, pattern color={rgb, 255:red, 0; green, 220; blue, 0}] (300,210) -- (330,210) -- (330,260) -- (300,260) -- cycle ;
\draw   (380,210) -- (410,210) -- (410,260) -- (380,260) -- cycle ;
\draw  [pattern=_lbm9x7hkc,pattern size=3pt,pattern thickness=0.75pt,pattern radius=0pt, pattern color={rgb, 255:red, 255; green, 163; blue, 0}] (340,210) -- (370,210) -- (370,260) -- (340,260) -- cycle ;
\draw [color={rgb, 255:red, 255; green, 0; blue, 0 }  ,draw opacity=1 ][line width=2.25]  [dash pattern={on 6.75pt off 4.5pt}]  (195,235) .. controls (199.41,285.96) and (228.31,286.98) .. (267.58,252.18) ;
\draw [shift={(270,250)}, rotate = 497.59] [color={rgb, 255:red, 255; green, 0; blue, 0 }  ,draw opacity=1 ][line width=2.25]    (17.49,-5.26) .. controls (11.12,-2.23) and (5.29,-0.48) .. (0,0) .. controls (5.29,0.48) and (11.12,2.23) .. (17.49,5.26)   ;

\draw [color={rgb, 255:red, 0; green, 0; blue, 255 }  ,draw opacity=1 ][line width=2.25]  [dash pattern={on 6.75pt off 4.5pt}]  (195,195) .. controls (362.09,137.18) and (431.7,177.32) .. (434.89,241.07) ;
\draw [shift={(435,245)}, rotate = 269.57] [color={rgb, 255:red, 0; green, 0; blue, 255 }  ,draw opacity=1 ][line width=2.25]    (17.49,-5.26) .. controls (11.12,-2.23) and (5.29,-0.48) .. (0,0) .. controls (5.29,0.48) and (11.12,2.23) .. (17.49,5.26)   ;

\draw [color={rgb, 255:red, 0; green, 220; blue, 0 }  ,draw opacity=1 ][line width=2.25]  [dash pattern={on 6.75pt off 4.5pt}]  (195,150) .. controls (278.79,133.34) and (313.12,168.54) .. (314.93,231.13) ;
\draw [shift={(315,235)}, rotate = 269.56] [color={rgb, 255:red, 0; green, 220; blue, 0 }  ,draw opacity=1 ][line width=2.25]    (17.49,-5.26) .. controls (11.12,-2.23) and (5.29,-0.48) .. (0,0) .. controls (5.29,0.48) and (11.12,2.23) .. (17.49,5.26)   ;

\draw [color={rgb, 255:red, 255; green, 163; blue, 0 }  ,draw opacity=1 ][line width=2.25]  [dash pattern={on 6.75pt off 4.5pt}]  (195,100) .. controls (322.89,81.38) and (353.3,166.48) .. (354.94,231.07) ;
\draw [shift={(355,235)}, rotate = 269.56] [color={rgb, 255:red, 255; green, 163; blue, 0 }  ,draw opacity=1 ][line width=2.25]    (17.49,-5.26) .. controls (11.12,-2.23) and (5.29,-0.48) .. (0,0) .. controls (5.29,0.48) and (11.12,2.23) .. (17.49,5.26)   ;

\draw (275,330) node   {$\Phi_{1/2} = \displaystyle \sum_i \sqrt{
	\textcolor[rgb]{1,0,0}{\frac{c(\pi_i)}{c(\psi_1)} \cdot \mfI_\pi(\pi_i ; \psi_1)} +
	\textcolor[rgb]{0,0,1}{ \frac{c(\pi_i)}{c(\psi_2)} \cdot \mfI_\pi(\pi_i ; \psi_2)} +
	\textcolor[rgb]{0,.86,0}{ \frac{c(\pi_i)}{c(\psi_3)} \cdot \mfI_\pi(\pi_i ; \psi_3)} +
	\textcolor[rgb]{1,0.64,0}{ \frac{c(\pi_i)}{c(\psi_4)} \cdot \mfI_\pi(\pi_i ; \psi_4)} + \ \ldots}$};

\end{tikzpicture}

%% file: freqalgo.tex
\section{Algorithm for $r$-bounded Instances}

\label{sec:rjuntas}

We can achieve a better approximation ratio if each function $g$ is an
$r$-junta for small $r$. Recall that an $r$-junta is a function that
depends on at most $r$ variables. In this section we prove the theorem:

\begin{theorem}
	There is a randomized algorithm that maintains an $r$-competitive solution in expectation to Fully Dynamic \subcov in the setting where functions arrive/depart over time, and these functions are each $r$-juntas. Furthermore this algorithm has total recourse:
	\[\frac{\sum_t g_t(\unvrs)}{\fmin}.\]
\end{theorem}

Our proof follows the framework of \cite{Gupta:2017:ODA:3055399.3055493}, which is itself a dynamic implementation of the algorithm of \cite{pitt1985simple}.

We start with some notation. Let $V_g = \{ u \in \unvrs \mid g(u) \neq 0 \}$ be the elements influencing $g$. Our assumption says that $|V_g| \leq r$. Let $S$ be the solution maintained by the algorithm. Each function $g \in G^{(t)}$ maintains a set $U_g \subseteq V_g$ of elements assigned to it. We say that $g$ is \textit{responsible} for $U_g$. We also define the following operation:

\textbf{Probing a function}.  
Sample one element $u \in V_g \backslash S$ with probability:
\[\frac{\displaystyle \frac{1}{c(u)} }{ \displaystyle \sum_{v \in V_g \backslash S} \frac{1}{c(v)} }.\]
Add the sampled element $u$ to the current solution $S$, and to $U_g$.

Given these definitions, we are ready to explain the dynamic algorithm.

\textbf{Function arrival}.
When a function $g$ arrives, initialize its element set $U_{g}$ to $\emptyset$. Then, while $g(S) \neq g(\unvrs)$, probe $g_t$.

\textbf{Function departure}. When a function $g$ departs, remove all its assigned elements $U_g$ from $S$. This may leave some set of functions $g_1, \ldots, g_s$ uncovered. For each of these functions $g_t$ in order of arrival, while $g_t(S) \neq g(\unvrs)$, probe $g_t$.

The $\sum_t g_t(\unvrs) / \fmin$ recourse bound is immediate, since the total number of probes can be at most $\sum_{t} g_t(\unvrs) / \fmin$ in total. It remains to bound the competitive ratio.

We prove the following:
\begin{lemma}
	\label{lem:rlemma}
	For any element $u\in \unvrs$:
	\[\expect{\sum_{\substack{g \in G^{(t)}: \\ u \in V_g}} \sum_{v \in U_g} c(v)} \leq r \cdot c(u). \]
\end{lemma}
This will imply as a consequence:
\[\expect{c(S)} \leq \sum_{o \in \opt_t}  \expect{\sum_{\substack{g \in G^{(t)}: \\ o \in V_g}} \sum_{v \in U_g} c(v)} \leq r \cdot \sum_{o \in \opt}  c(o) = r \cdot c(\opt_t).\]

\begin{proof}[Proof of \cref{lem:rlemma}]
	The proof is by induction. Fix $u \in \unvrs$, and consider the functions $g \in G^{(t)}$ for which $u \in V_g$. Let $X_i$ be the random variable that is the $i^{th}$ function probed. Let $Y_i$ be the random variable that is the element of $N$ sampled during the $i^{th}$ probe. With this notation, the inductive hypothesis is:
	\[\expect{\sum_{i \geq j} c(Y_i) \mid X_1, Y_1, \ldots X_{j-1}, Y_{j-1}}  \leq r \cdot c(u).\]
	For the base case $i=m$, note that given $X_1, Y_1, \ldots, X_{m-1}, Y_{m-1}$, the variable $X_m$ is determined. Suppose $X_m = g$. Then:
	\begin{align*}
	\expect{c(Y_m) \mid X_1, Y_1, \ldots X_{m-1}, Y_{m-1} } &= \expect{c(Y_m) \mid X_1, Y_1, \ldots X_{m-1}, Y_{m-1}, X_m = g} \\
	&= \sum_{v \in V_g} \frac{1}{ c(v) \left(\sum_{v' \in V_g \backslash S_t} \frac{1}{c(v')} \right)} \cdot c(v) \\
	&\leq \frac{|V_g|}{\sum_{v' \in V_g \backslash S_t} \frac{1}{c(v')}} \\
	&\leq r \cdot c(u).
	\end{align*}
	For the inductive step, suppose the claim holds for $j+1$, and consider the case for $j$:
	\begin{align*}
	&\expect{\sum_{i \geq j} c(Y_i) \mid X_1, Y_1, \ldots X_{j-1}, Y_{j-1}} \\
	&=\expect{ c(Y_j) \mid X_1, Y_1, \ldots X_{j-1}, Y_{j-1}} + \expect{\sum_{i \geq j+1} c(Y_i) \mid X_1, Y_1, \ldots X_{j-1}, Y_{j-1}} \\
	&\leq \frac{r}{\sum_{v' \in V_g \backslash S_t} \frac{1}{c(v')}} + \sum_{u' \neq u}  \expect{\sum_{i \geq j+1} c(Y_i) \mid X_1, \ldots Y_{j-1}, Y_j = u'} \cdot \prob{Y_j = u' \mid X_1, \ldots, Y_{j-1}} \\
	&\leq \frac{r}{\sum_{v' \in V_g \backslash S_t} \frac{1}{c(v')}} + r\cdot c(u) \cdot \left(1 - \frac{1}{c(u) \sum_{v' \in V_g \backslash S_t} \frac{1}{c(v')}}\right) \\
	&= r \cdot c(u). \qedhere
	\end{align*}
\end{proof}

%% file: combiner.tex
\section{Combiner Algorithm}
 
\label{sec:combiner}

We show that we can adapt the combiner algorithm of \cite{Gupta:2017:ODA:3055399.3055493} to our general problem.

\begin{theorem}
	Let $A_G$ be an $O(\log f(N) / \fmin)$-competitive algorithm for fully-dynamic \subcov with amortized recourse $R_G$. Let $A_{PD}$ be an $O(r)$-competitive algorithm for fully-dynamic \subcov when all functions are $r$-juntas with amortized recourse $R_{PD}$. Then there is an algorithm achieving an approximation ratio of $O(\min(\log (f(N)/ \fmin), r)$ for fully-dynamic \subcov when all functions are $r$-juntas, and it has total recourse $O(R_{G} + R_{PD})$.
\end{theorem}

\begin{proof}
	The idea is to partition the functions into different buckets based on their junta-arity, in powers of $2$ up to $\log f(N) / \fmin$. We run a copy of $A_{PD}$ which we call $A^{(\ell)}_{PD}$ on each bucket $B_\ell$ separately, and run $A_G$ one single time on the set of remaining functions.
			
	Formally, for every index $0 < \ell < \lceil \log \log (f(N) / \fmin) \rceil$, maintain a bucket $B_\ell$ representing the set of functions $g$ such that $g$ is a $k$-junta, for $k \in [2^\ell, 2^{\ell+1})$. Also maintain the bucket $B_G$ for any remaining functions. When a functions arrives, we insert it into exactly one appropriate bucket and update the appropriate algorithm.

	\begin{lemma}
		The total cost of the solution maintained by the algorithm is $O(\min(\log f(N), r)$.
	\end{lemma}
	
	\begin{proof}
		If $r \leq \log (f(N) / \fmin)$, the algorithm never runs $A_G$. Each algorithm $A^{(\ell)}_{PD}$ is $O(2^{\ell+1})$-competitive, and thus maintains a solution of cost no more than $O(2^{\ell + 1}) c(\opt)$. The largest bucket index is $\ell_{\max} = \lceil \log r \rceil$. Hence the total cost of the solution is:
		\[ \sum_{\ell=1}^{\ell_{\max}} O(2^{\ell + 1})\cdot c(\opt) = O(r) \cdot c(\opt).\]
		
		If on the other hand, $r > \log (f(N) / \fmin)$, then the largest bucket index is $\ell_{\max} = \lceil \log \log (f(N) / \fmin) \rceil$. The total cost of the $A_{PD}$ algorithms is then 
		\[ \sum_{\ell=1}^{\ell_{\max}} O(2^{\ell + 1})\cdot c(\opt) = O(\log f(N) / \fmin) \cdot c(\opt).\]
		Meanwhile, the total cost of the solution maintained by $A_G$ on the remaining functions has cost $O(\log f(N) / \fmin) \cdot c(\opt)$. Thus the global solution maintained by the combiner algorithm is also $O(\log f(N) / \fmin)$-competitive.
	\end{proof}
	
	The recourse bound is immediate since each function $g$ arrives to/departs from exactly one bucket, so at most one algorithm among $\{A^{(\ell)}_{PD}\}_\ell \cup \{A_G\}$ has to update its solution at every time step.
\end{proof}

%% file: beyondsubmodularity.tex
\section{Further Applications}

\label{sec:furtherapps}

In this section, we show how to recover several known results on recourse bounded algorithms using our framework. We hope this is a step towards unifying the theory of low recourse dynamic algorithms.

\subsection{Online Metric Minimum Spanning Tree}

\label{sec:mst}

In this problem, vertices in a metric space are added online to an active set. Let $A_t$ denote the active set at time $t$. After every arrival, the algorithm must add/remove edges to maintain a spanning tree $S_t$ for $A_t$ that is competitive with the \mst. We show:

\begin{theorem}
	\label{thm:mst}
	There is a deterministic algorithm for Online Metric \mst that achieves a competitive ratio of $O(1)$ and an amortized recourse bound of $O(\log D)$, where $D$ is the ratio of the maximum to minimum distance in the metric.
\end{theorem}
\cite{gu2016power} show how to get an $O(1)$ worst case recourse bound.

\mst is a special case of \subcov in which $\unvrs$ is the set of edges of the graph, and $f$ is the rank function of a graphic matroid. The main difference between the dynamic version of this problem and our setting is that here vertices arrive online \textit{along with all their incident edges}. Hence not only is the submodular function changing, but $\unvrs$ is also growing. We show that this detail can be handled easily.

We define the submodular function $f^{(t)}$ to be the rank function for the current graphic matroid, i.e. $f(S) = |A_t| - c_t$, where $c_t$ is the number of connected components induced by $S$ on the set of vertices seen thus far. Note that $\fmax = \fmin = 1$, so an edge having nonzero coverage is equivalent to the edge being in our current solution, $S_t$.

Now the algorithm is:
\begin{algorithm}[H]
	\caption{\textsc{FullyDynamicMST}}
	\label{alg:dynamicMST}
	\begin{algorithmic}[1]
		\State $\pi \leftarrow$ arbitrary initial permutation of edges.
		\For{$t = 1, 2, \ldots, T$}
		\State When vertex $v_t$ arrives, add edges incident to $v_t$ to tail of permutation in arbitrary order, and update $f^{(t)}$. 
		\While{there exists a legal $\gamma$-move or a swap for $\pi$}
		\State Perform the move, and update $\pi$.
		\EndWhile
		\State Output the collection of $\pi_i$ such that $\mff_\pi(\pi_i) > 0$.
		\EndFor 
	\end{algorithmic}
\end{algorithm}

As in \cref{cor:stacktrace}, if the algorithm terminates then it represents the stack trace of an approximate greedy algorithm for \mst. Hence the solution is $O(1)$ competitive. To bound the recourse, we use the general potential $\Phi_h$ from \cref{sec:weightedsubmod}. As before, local moves decrease the potential $\Phi_h$ by $\eps_\gamma\cdot \cmin \cdot h\left(\frac{1}{\cmin}\right)$, so it suffices to show that the potential does not increase by too much when $f^{(t)}$ is updated. Exactly one new edge will have increased marginal coverage, and its coverage will increase from $0$ to $1$. Thus the increase in potential is at most $\cmax \cdot h(1/\cmax)$. Together, these imply an amortized recourse bound of:
\[\frac{1}{\eps_{\gamma}} \cdot \frac{\cmax}{\cmin} \cdot \frac{h(1/\cmax)}{h(1/\cmin)}.\]
Setting $h(x) = x^{1-\delta}/(1-\delta)$ along with $\delta = (\ln (\cmax / \cmin + 1))^{-1}$ and $\gamma = \epsilon^2$ as in \cref{thm:main_weighted_subc}, we have $\eps_\gamma \geq \delta$, and hence we get a recourse bound of $O(\ln(\cmax / \cmin)) = O(\ln D)$.

\subsection{Fully-Dynamic Metric Minimum Steiner Tree}

We show that we can also fit into our framework the harder problem of maintaining a tree that spans a set of vertices in the fully-dynamic setting where vertices can both arrive and depart. We must produce a tree $S_t$ that spans the current set of active vertices $A_t$, but we allow ourselves to use Steiner vertices that are not in the active set. We show:

\begin{theorem}
	\label{thm:steiner}
	There is a deterministic algorithm for Fully-Dynamic Metric \mstt that achieves a competitive ratio of $O(1)$ and an amortized recourse bound of $O(\log D)$, where $D$ is the ratio of the maximum to minimum distance in the metric.
\end{theorem}
This guarantee matches that of \cite{lkacki2015power}. Separately \cite{gupta2014online} showed how to improve the bound to $O(1)$ amortized recourse.

Our algorithm is the same local search procedure as before, with one twist. We maintain a set of vertices $L$ we call the live set. This set is the union of the active terminals we need to span, and any Steiner vertices currently being used. We define $f^{(t)}$ similarly to before as $f^{(t)}(S) = |L| - c_t$, where $c_t$ is the number of connected components induced by the edge set $S$ on the set of vertices in $L$. Note that this function is submodular, because it is the rank function of the graphic matroid on the live vertex set $L$.

Now when a vertex $v$ departs, we mark it as a Steiner vertex but leave it in the live set. If at any point during the local search $\deg(v) = 2$, we replace $v$ with the edge that shortcuts between $v$'s two neighbors. If at any point point $\deg(v) = 1$, we delete $v$ and its neighboring edge.

\begin{algorithm}[ht]
	\caption{\textsc{FullyDynamicSteinerTree}}
	\label{alg:dynamicsteiner}
	\begin{algorithmic}[1]
		\State $\pi \leftarrow$ arbitrary initial permutation of edges.
		\For{$t = 1, 2, \ldots, T$}
		\If{vertex $v_t$ arrives}
			\State Add edges incident to $v_t$ to tail of permutation in arbitrary order, and update $f^{(t)}$.
		\ElsIf{vertex $v_t$ departs}
		\State Mark $v_t$ as a Steiner vertex. Run \textsc{CleanSteinerVertices}.
		\EndIf 
		\While{there exists a legal $\gamma$-move or a swap for $\pi$}
		\State Perform the move, and update $\pi$.
		\State Run \textsc{CleanSteinerVertices}.
		\EndWhile
		\State Output the collection of $\pi_i$ such that $\mff_\pi(\pi_i) > 0$.
		\EndFor 
	\end{algorithmic}
\end{algorithm}

\begin{algorithm}[ht]
	\label{alg:steinerhelper}
	\begin{algorithmic}[1]
		\Procedure{CleanSteinerVertices}{}
		\While{there is a Steiner vertex $v$ with $\deg(v) = 2$}
		\State Let $u_1$ and $u_2$ be the neighbors of $v$. 
		\State Add the edge $(u_1, u_2)$ to the position of $(v, u_1)$ in $\pi$. \algorithmiccomment{this shortcuts $v$}
		\State Delete all edges incident to $v$ from $\unvrs$, remove $v$ from the live set, and update $f^{(t)}$.
		\EndWhile
		\While{there is a Steiner vertex $v$ with $\deg(v) = 1$}
		\State Delete all edges incident to $v$ from $\unvrs$, remove $v$ from the live set, and update $f^{(t)}$.
		\EndWhile
		\EndProcedure
	\end{algorithmic}
\end{algorithm}

To show the competitive ratio we can rely on known results \cite{imase1991dynamic, gupta2014online}. If \cref{alg:dynamicsteiner} terminates, the output tree is known as a \textit{$\gamma$-stable extension tree} for the terminal set $S$.

\begin{lemma}[Lemma 5 of \cite{imase1991dynamic}]
	If $T$ is a $\gamma$-stable extension tree for $A_t$, then:
	\[c(T) \leq 4\gamma \cdot c(\opt(A_t)) \]
	where $\opt(A_t)$ is the optimal Steiner tree for terminal set $A_t$.
\end{lemma}

Since we set $\gamma = e^2$, this gives us a competitive ratio of $4 e^2 = O(1)$.

It remains to show the recourse bound. Deleting degree $1$ and $2$ vertices requires a constant number of edge changes, so this can be charged to each vertex's departure.  We show that the potential argument from before is not hampered by the changes to the algorithm.

\begin{claim}
	The procedure \textsc{CleanSteinerVertices} does not increase the potential.
\end{claim}

\begin{proof}
	When a degree $1$ Steiner vertex is deleted, the incident edge is removed from the permutation and no other edge's marginal coverage changes.
	
	When a degree $2$ Steiner vertex is deleted, the edges $(v, u_1)$ and $(v, u_2)$ are replaced by the edge $(u_1, u_2)$. Recall that our choice of potential is:
	\[\Phi_h(\pi) = \sum_{e \in S_t} c(e)^\delta\]
	for $0 < \delta < 1$.
	By triangle inequality, and concavity of $h$:
	\[d(u_1, u_2)^\delta \leq (d(v, u_1) + d(v, u_2))^\delta \leq d(v, u_1)^\delta + d(v, u_2)^\delta.\]
	Thus this replacement only decreases the potential.
\end{proof}

Otherwise, the potential increases during vertex arrivals and decreases during $\gamma$-moves exactly as in \cref{sec:mst}. We are left with the same recourse bound of $O(\ln D)$.

%% file: appendix.tex
\section{Bounds Using the Shannon Entropy Potential}

\label{sec:shannonentropy}

We show that Shannon entropy also works as a potential, albeit with the weaker recourse bound of:
\[O \left(\frac{\sum_t g_t(\unvrs)}{\fmin} \ln \left( \frac{\cmax}{\cmin}\cdot \frac{\fmax}{\fmin} \right)\right).\]

Define the Shannon Entropy potential to be the expression:
\[\Phi_1(f, \pi) := \sum_{i\in N} \mff_{\pi}(\pi_i) \log \frac{c(\pi_i)}{\mff_{\pi}(\pi_i)}.\]
In order to ensure that $\Phi_1$ remains nonnegative and monotone in each $\mff_{\pi}(\pi_i)$, scale $c$ by $1 / \cmin$ and $f$ by $1/(e \cdot \fmax)$ such that all costs are greater than $1$ and all coverages are less than $1/e$. We will account for this scaling at the end.

Note that $\Phi_1$ is $\Phi_h$ from \cref{sec:weightedsubmod} with $h(x) = x \log(1/x)$. This $h$ satisfies properties \labelcref{gencost_srom1,gencost_srom2,gencost_srom4} but not \labelcref{gencost_srom3}.

\begin{mdframed}
	\textbf{Properties of $\Phi_{1}$:}
	\begin{enumerate}[label=\textbf{(\Roman*)}]
		\item $\Phi_1$ increases by at most $g_t(\mathcal{N}) \cdot \ln \left( \cmax / \fmin \right)$ with the addition of function $g_t$ to the active set. \label{ent_rom1}
		\item $\Phi_1$ does not increase with deletion of functions from the system. \label{ent_rom2}
		\item $\Phi_1$ does not increase during swaps. \label{ent_rom3}
		\item If $\gamma> e$, then $\Phi_1$ decreases by at least $\fmin \ln(\gamma / e)$ with every $\gamma$-move. \label{ent_rom4}
	\end{enumerate}
\end{mdframed}

The proofs that $\Phi_1$ satisfies properties \labelcref{ent_rom1,ent_rom2,ent_rom3} follows directly from \cref{lem:gencost_hprops}, since these do not use property \labelcref{gencost_srom3}. It remains to show the last property.

\begin{lemma}
	If $\gamma > e$, every $\gamma$-move decreases $\Phi_1$ by at least $\fmin \cdot \ln(\gamma / \epsilon)$.
\end{lemma}

\begin{proof}
	
	Suppose we perform a $\gamma$-move on a permutation $\pi$. Let $u$ be the element
	moving to some position $p$ from some position $q > p$, and let $\pi'$ denote the permutation after the move. For convenience, also define:
	\begin{align*}
	v_i &:= \mff_\pi (\pi_i), \tag{the original coverage of the $i^{th}$ set}\\
	a_i &:= \mutcov_f(\pi_i; u \mid \pi_{1:i-1}) = \mff_\pi (\pi_i) - \mff_{\pi'}(\pi_i). \tag{the loss in coverage of the $i^{th}$ set}
	\end{align*}
	Then:	
	\begin{align}
	&\Phi_1(f, \pi') - \Phi_1(f, \pi) \notag \\
	&= \sum_{i=1}^n (v_i - a_i) \ln \frac{c(\pi_i)}{v_i - a_i} + \sum_{i=1}^n a_i \ln \frac{c(u)}{\sum_{i=1}^n a_i} \notag \\
	& \ \ \ - \sum_{i=1}^n v_i \ln \frac{c(\pi_i)}{v_i} \notag \\
	&\leq - \sum_{i=1}^n a_i \ln \left(\frac{c(\pi_i)}{e \cdot v_i}\right) + \sum_{i=1}^n a_i \ln \frac{c(u)}{\sum_{i=1}^n a_i} \label{ent:concave}\\
	&\leq - \sum_{i=1}^n a_i \ln \left(\frac{\gamma}{e} \cdot \frac{c(u)}{\sum_i a_i}\right) + \sum_{i=1}^n a_i \ln \frac{c(u)}{\sum_{i=1}^n a_i} \label{ent:gammamove}\\
	&= - \sum_{i=1}^n a_i \ln \left( \frac{\gamma}{e} \right) \notag\\
	&= - \fmin \cdot \ln \left( \frac{\gamma}{e} \right). \notag
	\end{align}
	
	Step \eqref{ent:concave} follows because, by concavity of the function $h(x) = x \log x$, we have $h(a+b) - h(a) \leq b \cdot h'(a)$. Step \eqref{ent:gammamove} follows because $u$ moving to position $p$ is a $\gamma$-move, hence $\sum_j a_j / c(u) \geq \gamma \cdot v_i / c(\pi_i)$.
\end{proof}

We now show the weaker recourse bound. By \labelcref{ent_rom1}, the potential $\Phi_h$ increases by at most 
\[g_t(\mathcal{N}) \cdot \ln \left( \frac{\cmax}{\fmin} \right) \]
with the addition of function $g_t$ to the active set. By \labelcref{ent_rom4}, it decreases by $\Omega(\fmin)$ with every move that costs recourse $1$, and otherwise does not increase. Since we scaled costs by $1/\cmin$ and coverages by $1/(e \cdot \fmax)$, this implies a recourse bound of:
\[O \left( \frac{\sum_t g_t(\unvrs)}{\fmin} \ln \left( \frac{\cmax}{\cmin}\cdot \frac{\fmax}{\fmin} \right) \right).\]